\documentclass[a4paper,twocolumn,11pt,accepted=2026-10-28, amssymb]{quantumarticle}
\pdfoutput=1
\usepackage[T1]{fontenc}
\usepackage[utf8]{inputenc}

\usepackage{xcolor}
\usepackage{graphicx}
\usepackage{braket}
\usepackage{caption}
\usepackage{subcaption}
\usepackage{xr}
\usepackage{amsthm}
\usepackage{dsfont}
\usepackage{float}
\usepackage{quantikz}
\usepackage[numbers,sort&compress]{natbib}

\usepackage{hyperref}
\hypersetup{
  colorlinks=true,
  linkcolor=blue,
  citecolor=blue,
  urlcolor=blue
}

\theoremstyle{definition}
\newtheorem{Def}{Definition}
\newtheorem{Prop}{Proposition}
\newtheorem{The}{Theorem}

\newtheorem{Cor}{Corollary}

\newcommand{\firstreview}[1]{\textcolor{black}{#1}}
\begin{document}
\title{Logarithmic-depth quantum state preparation of polynomials}
\author{Baptiste Claudon}
\email{baptiste.claudon@qubit-pharmaceuticals.com}
\affiliation{Qubit Pharmaceuticals, Advanced Research Department, 75014 Paris, France}
\affiliation{Sorbonne Universit\'e, LJLL, UMR 7198 CNRS, 75005 Paris, France}
\affiliation{Sorbonne Universit\'e, LCT, UMR 7616 CNRS, 75005 Paris, France}
\author{Alexis Lucas}
\affiliation{Sorbonne Universit\'e, LCT, UMR 7616 CNRS, 75005 Paris, France}
\author{Jean-Philip Piquemal}
\affiliation{Qubit Pharmaceuticals, Advanced Research Department, 75014 Paris, France}
\affiliation{Sorbonne Universit\'e, LCT, UMR 7616 CNRS, 75005 Paris, France}
\author{César Feniou}
\email{cesarf@qubit-pharmaceuticals.com}
\affiliation{Qubit Pharmaceuticals, Advanced Research Department, 75014 Paris, France}
\author{Julien Zylberman}
\email{julien.zylberman@cerfacs.fr}
\affiliation{CERFACS, Toulouse, France}
\begin{abstract}
Quantum state preparation is a central primitive in many quantum algorithms, yet it is generally resource intensive, with efficient constructions known only for structured families of states. This work introduces a method for preparing quantum states whose amplitudes are given by a degree$-d$ polynomial, using circuits with logarithmic depth in the number $n$ of qubits and only $\mathcal O(n)$ ancilla qubits, improving previous approaches that required linear-depth circuits. The construction first relies on a block-encoding of an affine diagonal operator based on its Pauli-basis decomposition, which involves only $n$ terms. A modified linear-combination-of-unitaries (LCU) technique is introduced to implement this decomposition in logarithmic depth, together with a novel circuit for the EXACT-one oracle that flags basis states in which exactly one qubit is in the state $\ket{1}$.
It then uses a generalized quantum eigenvalue transformation (GQET) to promote this affine operator to an arbitrary degree polynomial.
Theoretical analysis and numerical simulations are reported along with a proof-of-principle implementation on a Quantinuum System Model H2 trapped-ion quantum processor using $14$ qubits and more than $500$ primitive quantum gates. Because polynomial approximations are ubiquitous in scientific computing, this construction provides a scalable and resource-efficient approach to quantum state preparation, further improving the potential of quantum algorithms in fields such as chemistry, physics, engineering, and finance.
\end{abstract}

\maketitle

\section{Introduction}

Preparing an $n$-qubit quantum state with prescribed components, a process commonly referred to as \textit{quantum state preparation} or \textit{quantum data loading}, is a fundamental subroutine in many quantum algorithms, used both to initialise qubits in carefully chosen states and to perform linear combination of unitaries (LCU) via the \texttt{PREPARE-SELECT} procedure \cite{qsvt}. Since the space of $n$-qubit states has dimension $2^n$, the initialisation of arbitrary states generically requires exponential resources ~\cite{barenco1995elementary, vartiainen2004efficient, plesch2011quantum}, which can significantly limit the practical performance of quantum algorithms. Consequently, only structured states can be prepared efficiently, \textit{i.e.} states whose components are specified by a function $f$ that exhibits exploitable features such as regularity, smoothness, or sparsity. The quest for efficient quantum circuits to prepare such states has therefore been an active area of research in recent years, driven by applications across a wide range of domains, including quantum machine learning~\cite{biamonte2017quantum, zhao2021smooth, das2024role, han2025enqode}, quantum simulation in chemistry~\cite{berry2025rapid, feniou2024sparse, huggins2025efficient, baker2025universal}, physics~\cite{reac_diffu, wave_equation, madelung, zylberman2025trotter}, and stochastic simulations~\cite{grover2002creatingsuperpositionscorrespondefficiently, herbert2021classicalsamplingcircuitquantum, Montanaro_2015}. In many such contexts, the classical data to be loaded into a quantum register can be accurately approximated by low-degree polynomials.

In the following, we introduce a method for preparing quantum states on \(n\) qubits whose amplitudes are given by a polynomial function of the computational-basis index, namely
\[
\ket{p}\ \propto\ \sum_x p(x)\,\ket{x}, 
\qquad 
p(x)=\sum_{k=0}^{d} a_k x^k,
\]
where \(p(x)\) is a degree-\(d\) polynomial with real or complex coefficients \(a_k\).
Our main result is a construction with circuit depth\footnote{The depth of a quantum circuit is the number of layers of primitive quantum gates that can be executed in parallel. It therefore represents the minimum number of sequential time steps required to implement a given quantum operation, assuming parallel gate execution is allowed.} \(\mathcal O(d\log n)\) and circuit size\footnote{The size of a quantum circuit is defined as the total number of primitive quantum gates it contains.} \(\mathcal O(d n)\) while the number of ancillary qubits\footnote{Ancillary qubits (ancilla) are additional qubits used to assist in implementing a desired quantum operation, for instance by enabling parallelization.} scale linearly in $n$.

Most existing quantum algorithms that include a state-preparation step proceed by implementing an oracle \(O_f\) that coherently evaluates a function \(f\), followed by controlled rotations that convert this value into amplitude on an ancilla register:
\[
\ket{x}\ket{0}^{\otimes m}
\ \xrightarrow{\,O_f\,}\
\ket{x}\ket{\tilde f(x)}\ket0
\ \xrightarrow{\,\mathrm{}\,}\
\tilde f(x)\ket{x}\ket{0}^{\otimes m}+\ket{\perp}.
\]

where $\tilde f(x)\in[0,1]$ denotes a finite-precision encoding of \(f(x)\), and $\ket{\perp}$ ensures that the quantum state is normalized~\cite{grover2002creatingsuperpositionscorrespondefficiently, sanders2019black, bausch2022fast, wang2022quantum, rattew2022preparing, lin2026quantumalgorithms}.
In practice, such oracles are typically realised either via coherent arithmetic circuits~\cite{vedral1996quantum,bhaskar2016quantum, haner2018optimizing} or via quantum lookup-table techniques~\cite{gidney2019windowed, krishnakumar2019quantum, sunderhauf2024block}.

Alternative approaches avoid coherent arithmetic by directly synthesising the target amplitudes for restricted families of functions \(f\).
This includes, for instance, Fourier modes~\cite{moosa2024linear, rosenkranz2025quantum}, Walsh modes~\cite{zylberman2024efficient,zylberman2025efficient}, Gaussian states~\cite{kuklinski2025simpler}, Chebyshev-based constructions~\cite{rosenkranz2025quantum}, tensor network algorithms to implement high-dimensional functions~\cite{ballarin2025efficientquantumstatepreparation}, and sparse states~\cite{tubman2018postponing, quantum_dac, de_Veras_2022, feniou2024sparse}, among others~\cite{o2025quantum, guseynov2026efficient}.

Within the specific setting of polynomial amplitudes \(p(x)\), straightforward arithmetic-based constructions remain depth-linear, since the phase-kickback step typically requires an \(\Omega(n)\) sequence of qubit \(R_Z\)-type rotations, while non-arithmetic approaches either still incur linear depth together with linear ancilla overhead and circuit size~\cite{gonzalez2024efficient} (including efficient decompositions into Clifford+$T$~\cite{o2025quantum}), or reduce the number of ancillas to \( \mathcal O(\log n)\) at the cost of maintaining linear depth and incurring an additional logarithmic factor in circuit size~\cite{guseynov2026efficient}. Note that logarithmic depth has been previously proposed, but only for polynomials of \(\sin(x)\)~\cite{mcardle2022quantum} rather than directly of \(x\).

Our approach achieves logarithmic depth while keeping a qubit count in $\mathcal O(n)$, yielding an asymptotic depth improvement over prior state-preparation methods for polynomial amplitudes.

The strategy consists in preparing \(\ket{p}\propto\ \sum_x p(x)\,\ket{x}\) by applying a diagonal operator that weights each computational basis state \(\ket{x}\) according to the target polynomial \(p(x)\).
Concretely, we apply a block-encoding\footnote{Block-encoding is a technique that embeds a target operator, up to a normalisation factor, as a sub-block of a larger unitary acting on additional ancilla qubits. This is necessary because the target operator is non-unitary whereas quantum gates are unitary. More precisely, a $(n+m)$-qubit unitary $U$ such that $\braket{O^m|U|O^m}=A/\alpha$ is said to be a block-encoding of $A$, with normalisation $\alpha>0$.} of the diagonal operator
\[
D_p=\sum_x \frac{p(x)}{\alpha}\,\ket{x}\!\bra{x},
\]
for a known normalisation \(\alpha\), to the uniform superposition
\[
\ket{s}=H^{\otimes n}\ket{0}=\frac{1}{\sqrt{N}}\sum_x \ket{x},
\]
where $N=2^n$. Postselecting in the correct ancilla subspace gives
\[
D_p\ket{s}\ \propto\ \sum_x p(x)\ket{x},
\] i.e., the desired polynomial-amplitude state up to normalisation.

The construction of $D_p$ proceeds in two steps: we first obtain a block-encoding of a linear diagonal operator, and then use the Generalized Quantum Eigenvalue Transform (GQET) protocol \cite{gqsvt} to promote this linear construction to an arbitrary degree-\(d\) polynomial \(p(x)\), \[
\sum_x x \,|x\rangle\langle x|
\;\xrightarrow{\;\mathrm{GQET}\;}
\sum_x p(x)\,|x\rangle\langle x|.
\]

Note that such diagonal operators can be used beyond state preparation, in particular for quantum chemistry, quantum Monte Carlo methods, and in the simulation of partial differential equations ~\cite{chan2023grid, gqsp, feniou2025real,zylberman2025efficient, zylberman2025trotter}. 

The central technical contribution lies in the construction of a log-depth block-encoding of the linear diagonal operator, achieved through a modified version of the \texttt{PREPARE-SELECT} framework, by exploiting the Pauli basis decomposition \[\sum_x x\,|x\rangle\langle x| \propto \alpha \mathbb{I} + \sum_{k=1}^n \frac{Z_k}{2^k}.\]

The EXACT-one (or one-hot) oracle—flagging the Hamming-weight-one subspace of some of the ancilla qubits—constitutes a central subroutine of our state preparation algorithm. By revisiting the EXACT-one circuit decomposition, we derive a logarithmic-depth implementation requiring only two ancilla qubits. Although this refinement does not modify the overall asymptotic scaling of the state preparation procedure, it improves upon previous constructions of the EXACT-one oracle, which required \(\mathcal O(n)\) ancillae in the same depth regime~\cite{logdepth_hamming, Piroli_2024, zi2024shallowquantumcircuitimplementation}.

The manuscript is organized as follows. The first section summarizes the main contributions. We then present the circuit constructions, followed by numerical simulations that validate the asymptotic scaling of the proposed circuits. We also report a proof-of-principle demonstration on the Quantinuum System Model H2 trapped-ion quantum processor~\cite{Pino_2021}, involving 14 qubits and more than 500 primitive quantum gates. Finally, we discuss the implications and limitations of our approach. Technical proofs and supplementary derivations are provided in the Appendix.

\section{Contributions}
In this section, we summarize the main contributions of the present work. Let $n\ge 1$ denote the number of qubits, and define $B_n=\{0.x_1...x_n:x_1,...,x_n\in \{0, 1\}\}$ as the set of points corresponding to the uniform discretization of the interval $[0,1]$. The first key result is a quantum circuit that block-encodes the linear position operator, as summarized in the following theorem.
\begin{The}{(Informal)} We construct an exact block-encoding of the $n$-qubit operator
\begin{equation}
\sum_{x\in B_n}x\ket{x}\bra x,
\end{equation}
with a quantum circuit of $\mathcal O(\log(n))$ depth, $\mathcal O(n)$ size, using $\mathcal O(n)$ ancillae and a normalization $\Theta(1)$ asymptotically independent of $n$.
\label{the:informal_xketxbrax}
\end{The}

The formal version is presented in Theorem~\ref{the:formal_xketxbrax} of the Appendix.

Using the GQET protocol detailed below, one can construct an efficient block-encoding of the diagonal operator $\sum_{x\in B_n}p(x)\ket{x}\bra x$. When applied to the uniform superposition state $\ket{s}=(1/\sqrt N)\sum_x \ket{x}$, this block-encoding prepares the desired target state up to a normalization asymptotically independent of $n$, as written in the following theorem:

\begin{The}{(Informal)}
Let $\epsilon>0$ and $p$ be a degree-$d$ complex (suitably normalized) polynomial. Then, we can prepare
\begin{equation}
\ket{p}=\frac1{\|p\|_{2,N}}\sum_{x\in B_n}p(x)\ket x,
\end{equation}
 up to an error $\epsilon$ in two-norm with a quantum circuit of depth $\mathcal O(d\log(\frac{\|p'\|_2}{\|p\|_2}/\epsilon)\log(n))$, size $\mathcal O(nd\log(\frac{\|p'\|_2}{\|p\|_2}/\epsilon))$, ancillae $\mathcal O(n)$, and a probability of success $\Theta(1)$ asymptotically independent of $n$ and $\epsilon$, where $\|p\|_{2,N}=\sqrt{\sum_{x\in B_n} |p(x)|^2}$, $\|p\|_2=\sqrt{\int_0^1|p(x)|^2}$ and $p'$ is the derivative function of $p$. 
 \label{the:informal_state_prep}
\end{The}
The formal statement is provided in Theorem~\ref{thm:maintheorem} in the Appendix.

In addition to the logarithmic-depth quantum state preparation, we design novel quantum circuits for the EXACT-one operator, a key subroutine of our method, which determines whether an input bitstring has Hamming weight one. By leveraging the concept of conditionally clean ancillae~\cite{Khattar2025riseofconditionally}, our EXACT-one construction achieves logarithmic depth using only two ancilla qubits, whereas previous logarithmic-depth implementations required a linear number of ancilla qubits (in the number of input bits).

\begin{The}
For every positive integer $n\ge1$, define the function $f_n:\{0, 1\}^n\to\{0, 1\}$ by:
\begin{equation}
f_n(x)=\left\{\begin{aligned}
&1&\text{if $\sum_{i=1}^nx_i=1$,}\\
&0&\text{else.}
\end{aligned}\right.
\end{equation}
We can construct a circuit $O_n$ such that $O_n\ket{x}\ket0=\ket{x}\ket{f_n(x)}$ with $\mathcal O(\log(n))$ depth, $\mathcal O(n)$ size, using $2$ ancillae.
\end{The}

\section{Methods}
In this section, we sketch the construction and its resource scaling; a schematic view of the associated quantum circuit is presented in Figure \ref{fig:overall_circuit} and additional technical details are given in the Appendix.

\paragraph{\textbf{Block-encoding of the linear function.}} We first build a block-encoding of the linear diagonal operator acting on $n$ qubits
\begin{equation}
L_n=\frac1{1-2^{-n}}\sum_{x\in B_n}x\ket x\bra x ,
\end{equation}
where the factor $(1-2^{-n})^{-1}$ implies that the largest eigenvalue  of $L_n$ is one.
A key feature of the operator $L_n$ lies in its decomposition in the Pauli basis, which contains only $n$ terms, as given by:
\begin{equation}
1-2L_n \;=\; \frac{1}{1-2^{-n}}\sum_{k=1}^n \frac{Z_k}{2^k},
\label{eq:pauli_decomp_method}
\end{equation}
where $Z_k$ is the $Z$-Pauli operator $Z=\begin{pmatrix} 1&0 \\0&-1\end{pmatrix}$ acting on the $k-$th qubit. This decomposition reduces the problem to implementing a dyadically-weighted sum of single-qubit \(Z_k\) operators.

To do so, we consider a modified approach of the usual linear combination of unitaries\footnote{The linear combination of unitaries $\sum_{k=1}^K\lambda_kU_k$ is usually made of a PREPARE routine to load the $\lambda_k$ coefficients on $\lceil\log_2(K)\rceil$ ancilla qubits, a SELECT routine $\sum_{k=1}^K U_k\otimes \ket{k}\bra{k}$ that correspond to a block-diagonal unitary, and a last $\text{PREPARE}^\dagger$ operation to reset the ancilla in their $\ket 0$ state.}, where the \texttt{PREPARE} routine loads the dyadic weights $1/2^k$ into a one-hot control register in logarithmic depth $\mathcal O(\log(n))$. The \texttt{PREPARE} is composed of the two following steps.

\paragraph{(i) Constant-depth initialization \(\ket{\alpha_n}\).}
First, prepare the state
\begin{equation}
\ket{\alpha_n}=\bigotimes_{k=1}^n\left(\sqrt{1-\frac1{2^k+1}}\ket0+\sqrt{\frac1{2^k+1}}\ket1\right),
\label{eq:alpha_state_method}
\end{equation}
using a single parallel layer of single-qubit \(R_Y\) rotations (depth \(O(1)\), size \(O(n)\)).

\paragraph{(ii) Log-depth EXACT-one filtering and flagging.}
Let us remark that \(\ket{1/2^k}\equiv\ket{0\cdots 01_k0\cdots 0}\) corresponds to the computational basis states of Hamming weight one.
The orthogonal projection of \(\ket{\alpha_n}\) onto the span of \(\{\ket{1/2^k}\}_{k=1}^n\) is proportional to
\begin{equation}
\ket{\beta_n}=\frac1{\sqrt{1-2^{-n}}}\sum_{k=1}^{n}2^{-k/2}\ket{1/2^k}.
\label{eq:beta_state_method}
\end{equation}
The projection is operated using a \textit{EXACT-one} oracle that flips a flag qubit on basis states of Hamming weight one, which gives from $\ket{\alpha_n}$ the following state
\begin{equation}
\ket{\psi_n}=\sqrt{p_n}\ket{\beta_n}\ket{1}+\sqrt{1-p_n}\ket{\gamma_n}\ket{0},
\label{eq:psi_state_method}
\end{equation}
where \(\ket{\gamma_n}\perp \ket{\beta_n}\) and the overlap \(p_n\) is bounded below by $1/6$, a constant independent of \(n\) (see the Appendix).
In this construction, the one-hot flagging is implemented in depth \(O(\log n)\) and size \(O(n)\) using only two ancillas.
We denote by \(\Psi_n\) the full \texttt{PREPARE} unitary satisfying 
\begin{equation}
    \Psi_n\ket{0}=\ket{\psi_n}.
\end{equation}

Then, we implement a unary \texttt{SELECT} operator \(S_n\) that applies \(Z_k\) to the target register whenever the control register is in \(\ket{1/2^k}\), and only when the one-hot flag is set.
This realizes the dyadic-weighted Pauli sum in Eq.~\eqref{eq:pauli_decomp_method}.
Since the flag qubit cannot control multiple gates at the same time, we copy the flag qubit into $n$ zeroed ancilla qubit, with $n$ CNOT gates arranged in $\mathcal O(\log(n))$ layers (see Lemma A.1 of \cite{zylberman2024fast} for more details). Therefore, the \texttt{SELECT} construction has depth \(\mathcal O(\log n)\) and size \( \mathcal O(n)\).

A block-encoding of $1-2p_nL_n$ is obtained by sandwiching \texttt{SELECT} between \texttt{PREPARE} and its inverse, i.e.,
\begin{equation}
U_{\mathrm{aff}} \;=\; \Psi_n^\dagger\, S_n\, \Psi_n,
\end{equation}
and then projecting onto the all-zero ancilla state. This yields a block-encoding of \(1-2p_nL_n\), where $p_n=\Omega(1)$. The quantum circuit associated with $U_{\mathrm{aff}}$ is represented in Figure \ref{fig:overall_circuit}. Then, by introducing an additional ancilla qubit $\ket{a=0}$ to control $U_{\mathrm{aff}}$, one can construct a block-encoding of $L_n$:
\begin{equation}
\begin{split}
U_{\mathrm{linear}}
&= (H\otimes 1)\,\mathcal C(U_{\mathrm{aff}})\,(H\otimes 1),
\end{split}
\end{equation}
with $\bra{a=0} U_{\mathrm{linear}} \ket{a=0}=p_n L_n$. Here, $(H\otimes 1)$ denotes a Hadamard gate acting on the newly introduced ancilla qubit $\ket{a}$ and the identity acting on the remaining qubits on which $U_{\mathrm{aff}}$ operates.

Importantly, controlling $U_{\mathrm{aff}}$ does not compromise the logarithmic-depth construction. First, only the \texttt{SELECT} operator needs to be controlled,
\begin{equation}
\mathcal C(U_{\mathrm{aff}})=\Psi_n^\dagger\,\mathcal C(S_n)\,\Psi_n .
\end{equation}
Second, $\mathcal C(S_n)$ can be implemented straightforwardly by leveraging the flag qubit, which already controls the $Z$-Pauli gate in the unary \texttt{SELECT} operator. In particular, applying a CNOT gate
\begin{equation}
\mathcal C(X)=1_{\text{Flag}}\otimes\ket{a=0}\bra{a=0}+X_{\text{Flag}}\otimes\ket{a=1}\bra{a=1},
\end{equation}
followed by a NOT gate $X_{\text{Flag}}$, which allows the flag qubit to correctly mark the desired state, thereby implementing the controlled operation.

These results are summarized in the formal theorem presented in the Appendix (see Theorem~\ref{the:formal_xketxbrax}).

\begin{figure*}[t]
    \centering
    \includegraphics[
        width=1.1\textwidth,
        trim=2cm 16.8cm 0 1.5cm,
        clip
    ]{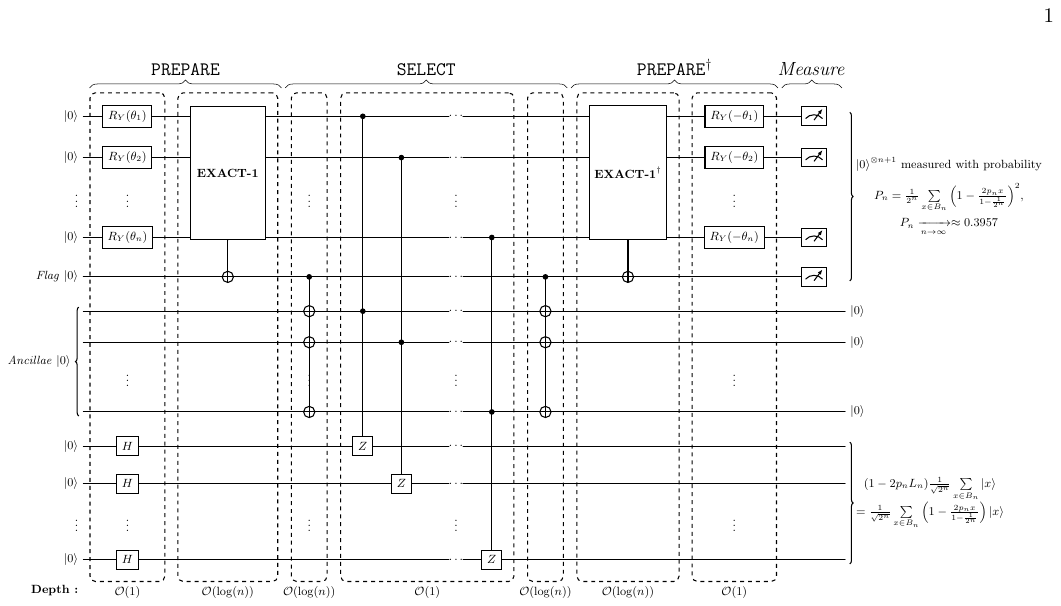}
    \caption{High-level structure of the affine block-encoding $U_{\text{aff}}$ construction of $1-2p_nL_n$ where $\theta_k = 2 \arcsin{\left(\sqrt{\frac{1}{1+2^k}}\right)}$. The circuit consists of \texttt{PREPARE} (\(\Psi_n\)), a unary \texttt{SELECT} (\(S_n\)) controlled on the one-hot subspace, and \(\texttt{PREPARE}^{\dagger}\), followed by projection (postselection) onto the designated ancilla outcome. Notice that the projection is not necessary when performing the amplitude amplification with $\mathcal O (1)$ steps and the GQET protocol.} 
    \label{fig:overall_circuit}
\end{figure*}

\paragraph{\textbf{Block-encoding polynomial functions.}}

\begin{figure*}[t]
\centering
\resizebox{0.8\textwidth}{!}{
\begin{quantikz}
   & \gate[1]{R(\theta_0, \phi_0^{'}, \lambda)} & \octrl{1} & \octrl{1} & \gate[1]{R(\theta_1, \phi_1^{'}, 0)} & \octrl{1} & \octrl{1} & \ldots & \gate[1]{R(\theta_{d-1}, \phi_{d-1}^{'}, 0)} & \octrl{1} & \octrl{1} & \gate[1]{R(\theta_{d}, \phi_{d}^{'}, 0)} & \qw &  \\
& \qwbundle{Ancillae} & \gate[2]{U} & \octrl{-1} & \qw & \gate[2]{U} & \octrl{-1} & \ldots & \qw & \gate[2]{U} & \octrl{-1} & \qw & \qw & \\
 &  & \qw& \qw& \qw & \qw & \qw & \ldots & \qw & \qw& \qw & \qw & \qw & 
\end{quantikz}
}
\caption{Quantum circuit for the generalized quantum eigenvalue transform (GQET) of a matrix $A$ block-encoded in a unitary $U=\begin{pmatrix}
    A/\alpha & * \\ * & *
\end{pmatrix}=(A/\alpha)\otimes \ket{0}\bra{0}+ \hdots$ by a polynomial $p(x)=\sum_{k=0}^da_kx^k=\sum_{k=0}^db_kT_k(x)$, where $T_k$ is the $k$-th Chebyshev polynomial of the first kind. The phase $(\{\theta_i\},\{\phi_i=\phi'_i-\pi\},\lambda)$ are those computed for the generalized quantum signal processing (GQSP) associated with the polynomial $P(z)=\sum_{k=0}^db_kz^k$. The rotation $R$ is defined as $R(\theta,\phi,\lambda)=\begin{pmatrix}
    e^{i(\lambda+\phi)}\cos(\theta) & e^{i\phi}\sin(\theta) \\ e^{i\lambda}\sin(\theta) & -\cos(\theta)
\end{pmatrix}$.}
\label{fig:quantum_circuit}
\end{figure*}
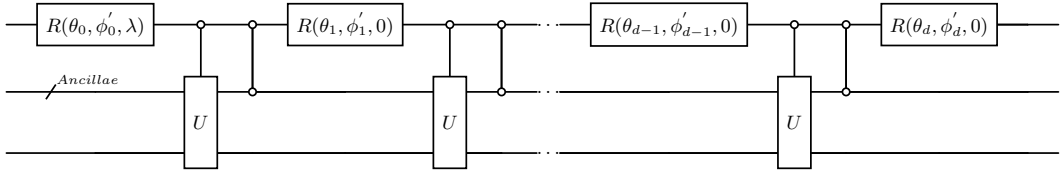

Thanks to the construction of a block-encoding of $L_n$, one can derive a block-encoding of a diagonal operator, whose eigenvalues are directly proportional to $p(x)$ by using the GQET protocol \cite{sunderhauf2024block} represented in Figure~\ref{fig:quantum_circuit}.
\firstreview{The GQET protocol extends the Generalized Quantum Signal Processing (GQSP) protocol~\cite{gqsp}, which enables the implementation of polynomials of unitaries, to the implementation of polynomials of block-encoded (or Projected Unitary Encoding) matrices. The protocol uses an additional ancilla qubit on which different rotations alternating with controlled versions of the block-encoding generate a linear combination of different powers of the block-encoded matrix. By appropriately choosing the rotation angles, the coefficients of this linear combination can be controlled, thereby enabling the implementation of a desired polynomial of the block-encoded matrix. Compared to the well-known Quantum Singular Value Transform (QSVT), the GQET protocol allows us to implement undefined-parity complex polynomials with easier angle computations during the circuit compilation. More details are presented in the Appendix and in reference \cite{sunderhauf2024block,gqsvt}.}

Applying the GQET protocol directly to $U_{\text{linear}}$ yields a polynomial in $p_n L_n$. Although one can construct a polynomial $q$ associated with the target polynomial $p$ such that $q(p_n L_n)\propto p(L_n)$, this approach results in a prohibitively small success probability, scaling as $\propto (p_n)^d$, where $d$ is the degree of the polynomial\footnote{This is due to the scaling factor of the polynomial, see Definition \ref{def:scalingfac} and Theorem \ref{the:gqet} in the Appendix for more details. }. To overcome this issue, the first step is to perform an amplitude amplification to obtain a block-encoding of
\begin{equation}
    (1-\epsilon)L_n,
\end{equation} 
with normalization one.
The amplitude amplification can be performed using standard techniques~\cite{Brassard_2002} or by using a GQET protocol where the associated polynomial approximates a Heaviside function up to an error \firstreview{$0\leq\epsilon\leq\delta$} with a polynomial degree \firstreview{$d=\mathcal O (\log(1/\delta))$} (more details are presented in the Appendix and in~\cite{Martyn_2021, eremenko2006uniformapproximationsgnxpolynomials}).

Then, the GQET protocol associated with the $p$ polynomial can be implemented to produce $p((1-\epsilon)L_n)$ which, by Rolle's theorem, gives an $\epsilon$-approximation of $p(L_n)$:
\begin{equation}
\|p((1-\epsilon)L_n)-p(L_n)\|_{\infty}\leq \epsilon\|p'\|_\infty.
\label{eq:pprimeinfty}
\end{equation}

The operator resulting from the GQET protocol\footnote{\firstreview{In practice, an additional source of error arises from the numerical computation of the GQET phases associated with a given polynomial. In the Appendix, we show that the resulting error in the quantum operation is bounded by the sum of the individual phase errors, and is therefore directly controlled by the accuracy of the angle-finding procedure. We further refer to existing angle-finding algorithms that have been shown to numerically achieve accuracies close to classical machine precision, even for high-degree polynomials.}} is the block-encoding of $\sum_{x\in B_n}p((1-\epsilon)x)\ket{x}\bra{x}$, with a normalization factor $\beta\|p\|_{\infty} $, where $\|p\|_{\infty} =\max_{x\in [0,1]}|p(x)|$ and $\beta$ is the scaling factor, defined in Equation \ref{def:scalingfac}, ensuring the GQET protocol to be well defined. 
Then, one can apply this operator on the full superposition state $\ket{s}=\left(1/2^{n/2}\right)\sum_{x\in B_n}\ket{x}$, giving, after success, an $\epsilon$-approximation of the target state $\ket{p}=\frac{1}{\|p\|_{2,N}}\sum_{x\in B_n}p(x)\ket x$. The probability of success being asymptotically independent of $n$ and $\epsilon$\footnote{The probability of success can scale with the polynomial degree $d$ according to the best known bound on the scaling factor $\beta=O(\log(d))$\cite{gqsvt}.}:
\begin{equation}
\begin{split}
     \mathbb P =&\left\|  \frac1{\beta\|p\|_{\infty}2^{n/2}}\sum_{x\in B_n}p((1-\epsilon)x)\ket x\right\|^2 \\&\xrightarrow{n\rightarrow +\infty, \textbf{ }\epsilon\rightarrow 0}\frac{\|p\|_2^2}{\beta^2\|p\|_{\infty}^2}=\Theta(1),
\end{split}
\end{equation}
with $\|p\|_2=\int_0^1|p(x)|^2dx$.

Notably, the GQET protocol preserves the logarithmic depth scaling since the different controlled-$U_{\text{linear}}$ can be implemented in $\mathcal{O}(1)$ gates by leveraging the flag qubit, as explained in the previous paragraph when controlling $U_{\text{aff}}$.
More details about these results and the corresponding quantum circuits are presented in the Appendix along with the main Theorem \ref{thm:maintheorem}.

\section{Implementation}

\begin{figure*}[t]
    \centering
    \begin{subfigure}[t]{0.48\textwidth}
        \vspace{0pt}
        \centering
        \includegraphics[width=\linewidth]{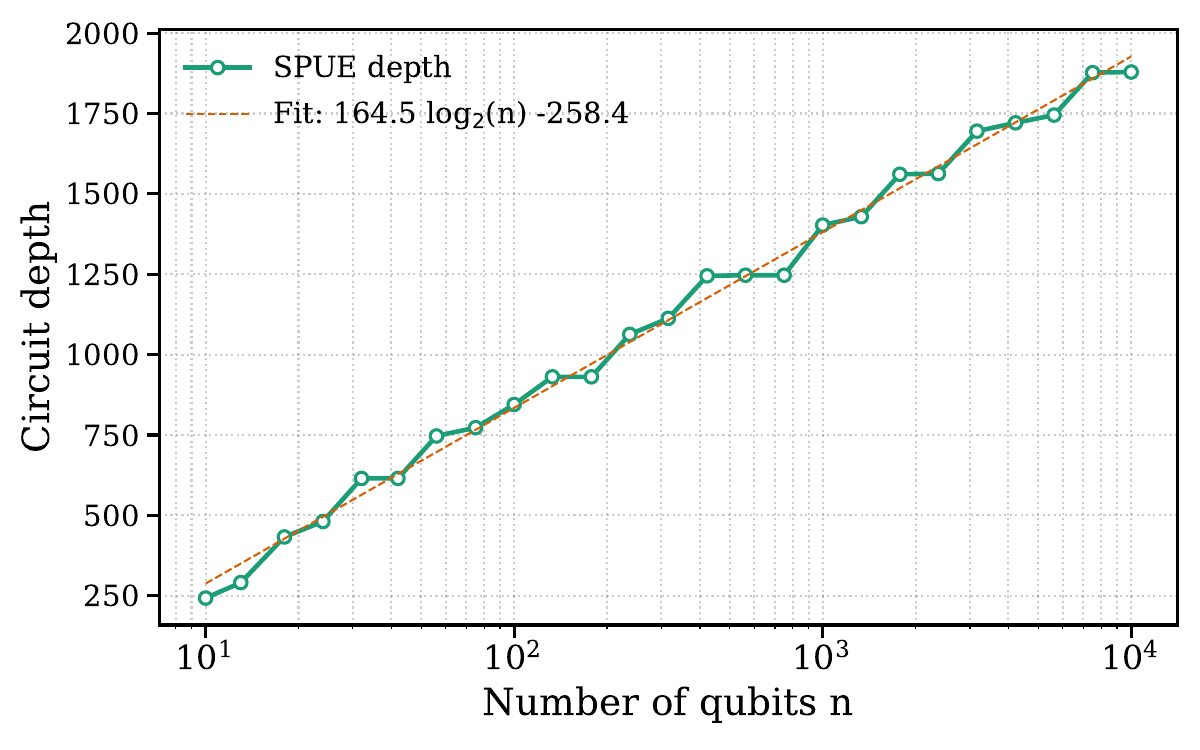}
        \caption{Circuit depth of the SPUE as a function of the number of qubits \(n\). The dashed line shows a fit of the form \(\alpha \log_2(n)+\beta\).}
        \label{fig:spue_depth}
    \end{subfigure}
    \hfill
    \begin{subfigure}[t]{.48\textwidth}
        \vspace{0pt}
        \centering
        \includegraphics[width=\linewidth]{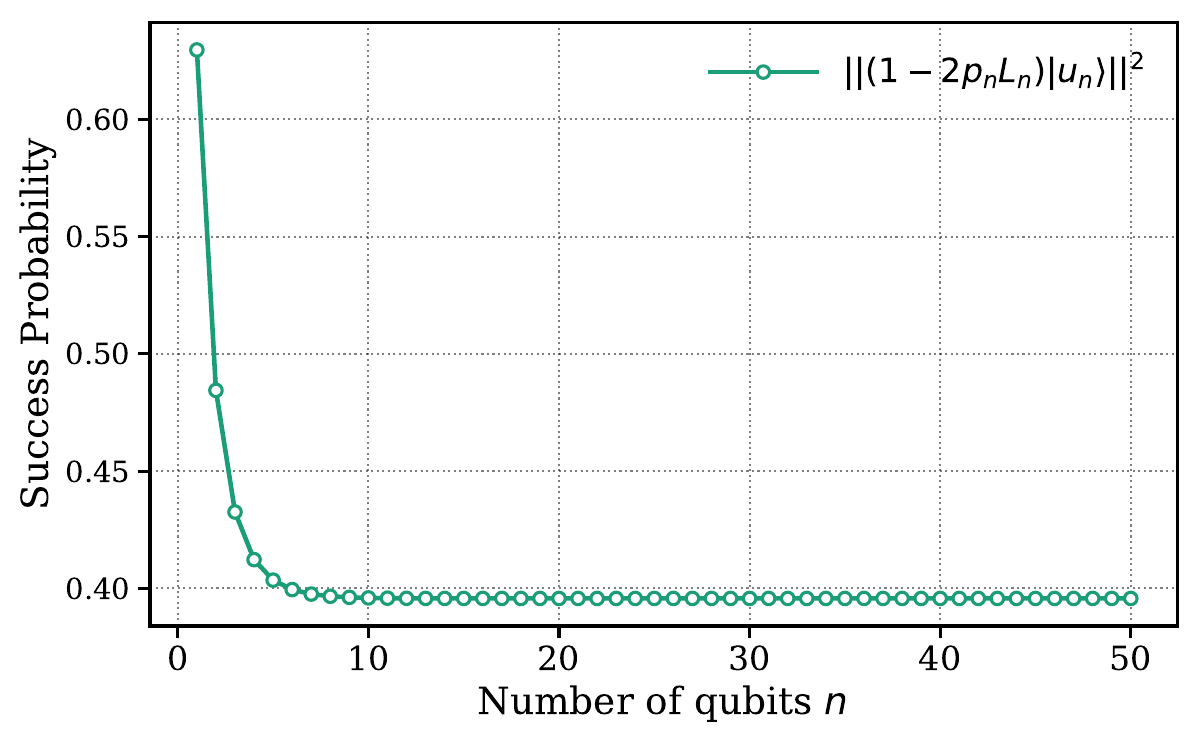}
        \caption{Convergence of the success probability \(P_n=\frac1{2^n}\sum_{x\in B_n}\left(1-\frac{2p_n x}{1-1/2^n}\right)^2\) as a function of \(n\), where $\ket{u_n} = \frac1{2^{n/2}}\sum_{x\in B_n} \ket x $.}
        \label{fig:pn_conv}
    \end{subfigure}
    \caption{Scaling behavior of the SPUE construction~\cite{codegit}.
    (\subref{fig:spue_depth}) Circuit depth grows logarithmically with the number of qubits.
    (\subref{fig:pn_conv}) The success probability \(P_n = \frac1{2^n} \sum\limits_{x \in B_n} \left(1-\frac{2p_nx}{1-\frac1{2^n}}\right)^2\) rapidly converges to a constant value \( P  \approx 0.3957\) .}
    \label{fig:spue_scaling}
\end{figure*}

We numerically validated the asymptotic scaling of the proposed block-encoding construction using circuit-level simulations. In particular, we benchmarked the circuit depth as a function of the number of qubits \(n\) and estimated the corresponding postselection success probability \(p_n\). As shown in Fig.~\ref{fig:spue_scaling}, the depth follows a logarithmic trend consistent with a fit of the form \(\alpha \log_2(n)+\beta\), while the success probability rapidly converges to a constant value as \(n\) increases. 

We additionally tested the construction on Quantinuum System Model H2 trapped-ion quantum processor by preparing a five-qubit state whose amplitudes follow a degree-one polynomial profile, and measuring the output distribution in the computational basis. Figure~\ref{fig:experimental_state_preparation} reports the resulting amplitudes (estimated as the square roots of the measured frequencies after postselection). The compiled implementation used \(14\) qubits in total, comprising \(274\) \texttt{PhasedX} gates, \(226\) \texttt{ZZPhase} gates, and \(11\) \(R_Z\) gates. We observe close agreement between the measured and target distributions, consistent with the reported single- and two-qubit gate fidelities of the Quantinuum System Model H2 device. 

\begin{figure}[H]
    \centering
    \includegraphics[width=\linewidth]{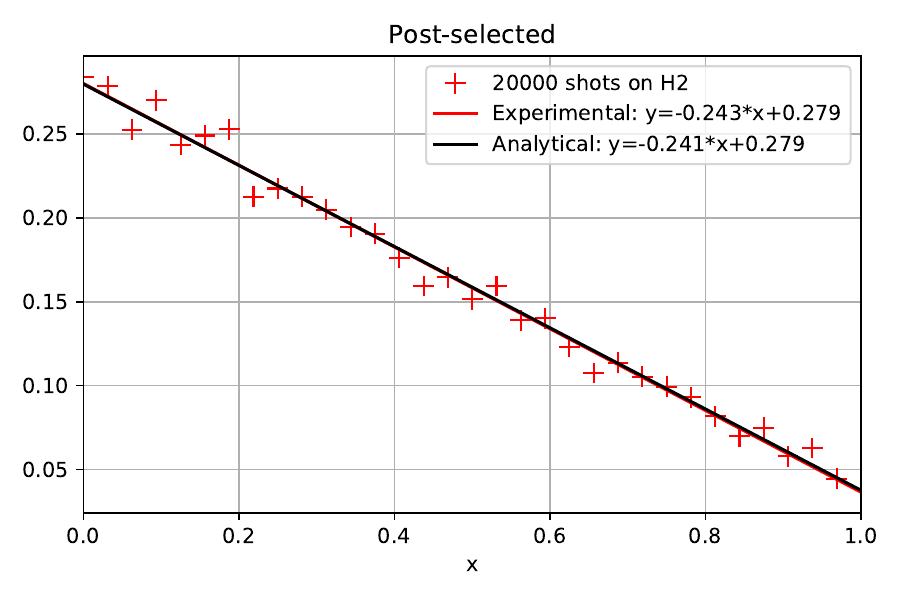}
    \caption{Hardware demonstration on Quantinuum System Model H2 of a degree-one polynomial-amplitude state preparation. The plotted values correspond to the square roots of the postselected measurement frequencies of each computational-basis outcome.}
    \label{fig:experimental_state_preparation}
\end{figure}

\section{Discussion and conclusion}
This study was motivated by the observation that, despite the ubiquity of structured amplitude/phase profiles, there appears to be no \emph{sublinear-depth} method for preparing quantum states as structured as \(\sum_x x\,\ket{x}\) or, more generally, any affine function of the computational basis index.

Existing constructions fall into two broad categories, both seemingly in linear depth.
\begin{enumerate}
  \item[(i)] \emph{Coherent arithmetic:} compute the affine function of the index using coherent arithmetic with \(O(\log n)\) ancillae, then propagate the result into phases. This typically entails an \(\Omega(n)\) depth overhead.
  \item[(ii)] \emph{Non-arithmetic methods:} other state-preparation strategies either lead to linear-depth circuits in known realizations, or implement only polynomials of \(\sin(x)\) (rather than of \(x\)) in logarithmic depth~\cite{mcardle2022quantum}.
\end{enumerate}

\firstreview{In the following, we provide a comparison of different quantum state preparation protocols which is summarized in Table~\ref{tab:costs_summary}. 
In~\cite{mcardle2022quantum}, the authors prepare a diagonal operator encoding the $\sin$ function before applying a polynomial approximation of $f\circ\sin^{-1}$ via a QSVT method. More precisely, 
if $f$ is a smooth function and $f\circ\sin^{-1}$ can be approached by a fixed-parity polynomial $h$ of degree $d$ on the interval $[-\sin(1), \sin(1)]$ with $\max_{x\in[-1, 1]}|h(x)|\leq 1$, the method presented in~\cite{mcardle2022quantum} presents a $\mathcal O(n)$ reduction in size and $\mathcal O(nd)$ reduction in space over the arithmetic-based methods. Compared to our method, if $f$ can be approached by a degree-$d$ polynomial over $[-1, 1]$, our state preparation has depth $\tilde{\mathcal O}(\log(n)d)$, at the expense of $\Omega(n)$ ancillae. It should be noted that the hypotheses for the two methods are not equivalent. For instance, the monomial $f(x)=x^m$ can be approached to fixed precision by a polynomial of degree $\mathcal O\left(\sqrt m\right)$ (see~\cite{saibaba2021approximatingmonomialsusingchebyshev}). As numerically shown in a notebook provided in~\cite{codegit}, approximating $f(y)=\left(\sin^{-1}(y)\right)^m$ requires a polynomial of degree $\Omega(m)$. Figure~\ref{fig:log-Max-Error_Versus_Degree}, in the Appendix, shows that the monomial $x^{20}$ is approached by polynomials in $x$ more efficiently than by polynomials in $\sin(x)$. Moreover, the condition $\max_{x\in[-1, 1]}|h(x)|\leq 1$ may be constraining for specific examples since, for the previous one, the best polynomial seems to explode on $[-1, 1]\backslash[-\sin(1), \sin(1)]$, leading to additional normalization costs. Conversely, when considering other commonly used functions such as the normal and binomial distributions, numerical evaluations show that the required polynomial degrees to achieve a fixed precision are nearly identical for both approaches, yielding very similar performance between the two methods ~\cite{codegit}. This discussion shows that our method and that of~\cite{mcardle2022quantum} are not equivalent in general. Our method is more efficient when the target function is well approximated by polynomials in $x$, whereas theirs is more efficient when the function is better approximated by polynomials in $\sin(x)$. }

\firstreview{In \cite{rattew2022preparing}, a quantum state preparation protocol is presented for continuous real-valued functions $f$ through an adiabatic Hamiltonian evolution. Starting from the initial Hamiltonian $H_0 = \ket{\psi_0}\bra{\psi_0}$, where $\ket{\psi_0} = H^{\otimes n}\ket{0}$, to the target Hamiltonian $H_1 = \ket{f}\bra{f}$, the adiabatic evolution is implemented via a sequence of small Trotterization steps. The resulting query complexity required to prepare the target quantum state with error $\epsilon$ scales as $O(\mathcal{F}^p/\epsilon^2)$, where $\mathcal{F}$ is a constant depending on the hypotheses of $f$ and $p=-4$. Each oracle can be implemented with $O(\text{poly}(n,d))$ gates and the number of ancillary qubits scales as $O(n+d)$, where $d$ denotes the number of digits used in the discretized representation of $f$. The claimed logarithmic complexity is with respect to the number of components $N=2^n$ in the $n$-qubit state $\ket{f}$, whereas our approach achieves an exponential reduction with a circuit depth that scales logarithmically with the number of qubits and logarithmically with the inverse of the target error, i.e., as $O(\log(n) \log(1/\epsilon))$.}

\firstreview{In~\cite{qalgapproxfuncload}, the authors prove that, under suitable smoothness assumptions, $\mathcal{O}\left(2^{k_0(\epsilon)}\right)$ gates are sufficient to prepare $\ket{f}$, where $k_0(\epsilon)=\max\left(2,-\frac{1}{2}\log_2\left(4^{-n}-\frac{96}{\eta^2}\log(1-\epsilon)\right)\right)$. Asymptotically, this results in a gate complexity scaling as $\mathcal{O}(1/\sqrt{\epsilon})$, whereas our method achieves both circuit depth and size scaling with $\epsilon$ as $\mathcal{O}(\log(1/\epsilon))$. Moreover, the smoothness condition imposed in~\cite{qalgapproxfuncload} requires $\eta=\sup_{x\in[0,1]}\left|\left(\frac{d}{dx}\right)^2\log\left(f^2(x)\right)\right|<\infty$, which is not satisfied, for instance, by the simple function $f(x)=x$.}

\firstreview{In~\cite{rattew2023nonlineartransformationsquantumamplitudes}, a protocol is introduced to transform a state-preparation unitary $U_f\ket{0}^{\otimes n}=\ket{f}$ into a diagonal block-encoding of the amplitudes of $\ket{f}$. Quantum singular value transformation (QSVT) is then used to apply a polynomial transformation to these amplitudes. This framework enables the preparation of a quantum state encoding a continuous function $f$, starting from a state-preparation oracle for an affine function. Since the protocol in~\cite{rattew2023nonlineartransformationsquantumamplitudes} assumes access to a state-preparation oracle, $O(1)$ calls to this oracle and an additional quantum circuit of depth $O(n)$ are required to construct the diagonal block-encoding. In contrast, our approach directly constructs a $O(\log(n))$-depth block-encoding of the affine function prior to applying the polynomial transformation, thereby avoiding these oracle calls and resulting in a lower circuit depth. Table~\ref{tab:costs_summary} summarizes the gate counts, depths and memory requirements of the above end-to-end methods.
}

\firstreview{
\begin{table*}
    \centering
    \begin{tabular}{|c|c|c|c|}
    \hline
\firstreview{Method}                                          &   \firstreview{Size}                                                                                        &   \firstreview{Depth}   & \firstreview{Ancillae}\\
\hline
Theorem~\ref{the:informal_state_prep}           &$\mathcal O\left(nd\log\left(\frac{\left\|p'\right\|_2}{\left\|p\right\|_2}/\epsilon\right)\right)$       &$\mathcal O\left(d\log\left(\frac{\left\|p'\right\|_2}{\left\|p\right\|_2}/\epsilon\right)\log(n)\right)$& $\Omega(n)$\\
\hline
Based on $\sin(x)$~\cite{mcardle2022quantum}    &$\mathcal O\left(nd_\epsilon\right/\mathcal F)$                                                           &$\mathcal O(nd_\epsilon/\mathcal F)$ & $3$\\
\hline
Grover-Rudolph without ancilla~\cite{qalgapproxfuncload}          &$2^{-\log\left(4^{-n}-\log(1-\epsilon)/\eta^2\right)/2}$                                                    &$2^{-\log\left(4^{-n}-\log(1-\epsilon)/\eta^2\right)/2}$                                                    & 0 \\
\hline
Adiabatic-Hamiltonian~\cite{rattew2022preparing}          &$O(\text{poly}(n)\mathcal{F}^{-4}/\epsilon^2)$                                                    &$O(\text{poly}(n)\mathcal{F}^{-4}/\epsilon^2)$                                                   & $\Omega(n)$ \\
\hline
    \end{tabular}
    \caption{Costs of end-to-end state preparation methods of $\ket f$. $d_\epsilon$ is the degree of the polynomial used to approximate $f\circ\sin^{-1}$ to a sufficient precision, $\mathcal F$ is a constant depending on the function $f$ (related to the filling ratio, see~\cite{mcardle2022quantum}), and $\eta=\sup_{x\in[0,1]}\left|\left(\frac d{dx}\right)^2\log\left(f^2(x)\right)\right|$. Space-time trade-offs are also possible, allowing for logarithmic depth scaling with the number of qubits~\cite{mcardle2022quantum}.}
    \label{tab:costs_summary}
\end{table*}
}

\paragraph{\textbf{Logarithmic-depth polynomial state preparation.}} We introduce a construction that prepares the target affine state with \emph{linear} ancilla count and overall circuit size, while achieving \emph{logarithmic depth}, the first log-depth method for affine profiles to the best of our knowledge.
Moreover, the extension from affine functions to general polynomials is immediate via GQET, at an additional overhead scaling with the polynomial degree and the logarithm of the inverse of the target error.

\paragraph{\textbf{A compact EXACT-one oracle.}}\leavevmode\\
A key ingredient is an explicit implementation of the EXACT-one, also known as hot-one, oracle, which flags computational basis states of Hamming weight exactly one.
Our construction uses only \emph{two ancillae} while retaining logarithmic depth, significantly improving over prior art in the same depth regime that required \emph{linear} ancilla overhead. \firstreview{Generalizing our construction to EXACT-$k$ operators for $k\geq2$ is not straightforward, as our construction relies on the specific structure of the EXACT-one operator. Nevertheless, other techniques exploiting conditionally clean ancilla qubits may enable similar speedups for flagging basis states with Hamming weight $1\leq k\leq n$ using an EXACT-$k$ gate. We leave the question of how to implement such gates optimally for future work.}

\paragraph{\textbf{Practical considerations and assumptions.}}\leavevmode\\
(i) \emph{Gate set:}
The first operation in \texttt{PREPARE} assumes access to \(R_Y\) rotations with arbitrary angles, including exponentially small angles.
When restricted to a finite, fault-tolerant gate set, approximating such rotations may introduce non-negligible synthesis overhead and can effectively restore linear depth in practice.
That said, \(R_Z\) rotations are often implemented \emph{virtually} on hardware (e.g., via frame updates~\cite{Pino_2021, helios}), which can substantially mitigate this concern and may yield experimentally best-in-class implementations. \firstreview{Finite-gate-set approximation and hardware errors can also be taken into account at each step of the proofs of our theorems.}
\newline (ii) \emph{Parallelism and connectivity:} the effort to compress quantum circuits and minimize depth relies on the assumptions of ideal parallel execution of quantum operations and full connectivity of the qubits. In fault-tolerant settings, however, constraints arising from gate synthesis (e.g., T-gate compilation), scheduling overhead, and realistic architecture may limit the effective parallelism, as discussed in \cite{claudon2024polylogarithmic, kuklinski2025simpler}.

Overall, this novel logarithmic-depth quantum state preparation provides an efficient framework for a wide range of applications in physics, engineering, chemistry, finance, quantum machine learning and related disciplines. In many such contexts, the classical data can be represented to high accuracy by low-degree polynomials, or more generally by functions that admit efficient polynomial approximations. By leveraging this structure, our approach enables scalable and resource-efficient state preparation, thereby expanding the range of practically accessible quantum algorithms. Future work should address optimal-depth state preparation under specific connectivity, a native gate set and an extension to multivariate polynomials.

\section{Acknowledgments}
This work has received funding from the European Research Council (ERC) under the European Union's Horizon 2020 research and innovation program (grant agreement No 810367), project EMC2 (JPP). Computations on the Quantinuum System Model H2 machine have been made possible through the Quantinuum Startup Partner Program. This work has also been supported by HQI, which is part of France 2030 under the French National Research Agency award number “ANR-22-PNCQ-0002”.

\section{Author contribution}

The authors acknowledge the use of large language models solely for spell-checking.

\bibliographystyle{plainnat}
\bibliography{bib}

\begin{thebibliography}{65}
\providecommand{\natexlab}[1]{#1}
\providecommand{\url}[1]{\texttt{#1}}
\expandafter\ifx\csname urlstyle\endcsname\relax
  \providecommand{\doi}[1]{doi: #1}\else
  \providecommand{\doi}{doi: \begingroup \urlstyle{rm}\Url}\fi

\bibitem[Baker et~al.(2025)Baker, Saxena, and Kyaw]{baker2025universal}
Jack~S Baker, Gaurav Saxena, and Thi~Ha Kyaw.
\newblock Universal initial state preparation for first quantized quantum simulations.
\newblock \emph{arXiv preprint arXiv:2510.07278}, 2025.

\bibitem[Ballarin et~al.(2025)Ballarin, García-Ripoll, Hayes, and Lubasch]{ballarin2025efficientquantumstatepreparation}
Marco Ballarin, Juan~José García-Ripoll, David Hayes, and Michael Lubasch.
\newblock Efficient quantum state preparation of multivariate functions using tensor networks, 2025.
\newblock URL \url{https://arxiv.org/abs/2511.15674}.

\bibitem[Barenco et~al.(1995)Barenco, Bennett, Cleve, DiVincenzo, Margolus, Shor, Sleator, Smolin, and Weinfurter]{barenco1995elementary}
Adriano Barenco, Charles~H Bennett, Richard Cleve, David~P DiVincenzo, Norman Margolus, Peter Shor, Tycho Sleator, John~A Smolin, and Harald Weinfurter.
\newblock Elementary gates for quantum computation.
\newblock \emph{Physical review A}, 52\penalty0 (5):\penalty0 3457, 1995.
\newblock \doi{https://doi.org/10.1103/PhysRevA.52.3457}.

\bibitem[Berntson and S{\"u}nderhauf(2025)]{berntson2025complementary}
Bjorn~K Berntson and Christoph S{\"u}nderhauf.
\newblock Complementary polynomials in quantum signal processing: Bk berntson, c. s{\"u}nderhauf.
\newblock \emph{Communications in Mathematical Physics}, 406\penalty0 (7):\penalty0 161, 2025.
\newblock \doi{https://doi.org/10.1007/s00220-025-05302-9}.

\bibitem[Berry et~al.(2025)Berry, Tong, Khattar, White, Kim, Low, Boixo, Ding, Lin, Lee, Chan, Babbush, and Rubin]{berry2025rapid}
Dominic~W. Berry, Yu~Tong, Tanuj Khattar, Alec White, Tae~In Kim, Guang~Hao Low, Sergio Boixo, Zhiyan Ding, Lin Lin, Seunghoon Lee, Garnet Kin-Lic Chan, Ryan Babbush, and Nicholas~C. Rubin.
\newblock Rapid initial-state preparation for the quantum simulation of strongly correlated molecules.
\newblock \emph{PRX Quantum}, 6:\penalty0 020327, May 2025.
\newblock \doi{10.1103/PRXQuantum.6.020327}.
\newblock URL \url{https://link.aps.org/doi/10.1103/PRXQuantum.6.020327}.

\bibitem[Bhaskar et~al.(2016)Bhaskar, Hadfield, Papageorgiou, and Petras]{bhaskar2016quantum}
Mihir~K. Bhaskar, Stuart Hadfield, Anargyros Papageorgiou, and Iasonas Petras.
\newblock Quantum algorithms and circuits for scientific computing.
\newblock \emph{Quantum Info. Comput.}, 16\penalty0 (3–4):\penalty0 197–236, March 2016.
\newblock ISSN 1533-7146.

\bibitem[Biamonte et~al.(2017)Biamonte, Wittek, Pancotti, Rebentrost, Wiebe, and Lloyd]{biamonte2017quantum}
Jacob Biamonte, Peter Wittek, Nicola Pancotti, Patrick Rebentrost, Nathan Wiebe, and Seth Lloyd.
\newblock Quantum machine learning.
\newblock \emph{Nature}, 549\penalty0 (7671):\penalty0 195--202, 2017.
\newblock \doi{10.1038/nature23474}.
\newblock URL \url{https://doi.org/10.1038/nature23474}.

\bibitem[Brassard et~al.(2002)Brassard, Høyer, Mosca, and Tapp]{Brassard_2002}
Gilles Brassard, Peter Høyer, Michele Mosca, and Alain Tapp.
\newblock Quantum amplitude amplification and estimation, 2002.
\newblock ISSN 0271-4132.
\newblock URL \url{http://dx.doi.org/10.1090/conm/305/05215}.

\bibitem[Chan et~al.(2023)Chan, Meister, Jones, Tew, and Benjamin]{chan2023grid}
Hans Hon~Sang Chan, Richard Meister, Tyson Jones, David~P. Tew, and Simon~C. Benjamin.
\newblock Grid-based methods for chemistry simulations on a quantum computer.
\newblock \emph{Science Advances}, 9\penalty0 (9):\penalty0 eabo7484, 2023.
\newblock \doi{10.1126/sciadv.abo7484}.
\newblock URL \url{https://www.science.org/doi/abs/10.1126/sciadv.abo7484}.

\bibitem[Claudon et~al.(2024)Claudon, Zylberman, Feniou, Debbasch, Peruzzo, and Piquemal]{claudon2024polylogarithmic}
Baptiste Claudon, Julien Zylberman, C{\'e}sar Feniou, Fabrice Debbasch, Alberto Peruzzo, and Jean-Philip Piquemal.
\newblock Polylogarithmic-depth controlled-not gates without ancilla qubits.
\newblock \emph{Nature Communications}, 15\penalty0 (1):\penalty0 5886, 2024.
\newblock \doi{10.1038/s41467-024-50065-x}.
\newblock URL \url{https://doi.org/10.1038/s41467-024-50065-x}.

\bibitem[Claudon et~al.(2026)Claudon, Feniou, Zylberman, and Lucas]{codegit}
Baptiste Claudon, César Feniou, Julien Zylberman, and Alexis Lucas.
\newblock Data for "logarithmic-depth quantum state preparation of polynomials" (https://zenodo.org/records/19064460).
\newblock 2026.
\newblock \doi{10.5281/zenodo.19064459}.
\newblock URL \url{https://zenodo.org/records/19064460}.

\bibitem[Costa et~al.(2019)Costa, Jordan, and Ostrander]{wave_equation}
Pedro C.~S. Costa, Stephen Jordan, and Aaron Ostrander.
\newblock Quantum algorithm for simulating the wave equation.
\newblock \emph{Phys. Rev. A}, 99:\penalty0 012323, Jan 2019.
\newblock \doi{10.1103/PhysRevA.99.012323}.
\newblock URL \url{https://link.aps.org/doi/10.1103/PhysRevA.99.012323}.

\bibitem[Das et~al.(2024)Das, Martina, and Caruso]{das2024role}
Sreetama Das, Stefano Martina, and Filippo Caruso.
\newblock The role of data embedding in equivariant quantum convolutional neural networks.
\newblock \emph{Quantum Machine Intelligence}, 6\penalty0 (2):\penalty0 82, 2024.
\newblock \doi{10.1007/s42484-024-00215-7}.
\newblock URL \url{https://doi.org/10.1007/s42484-024-00215-7}.

\bibitem[de~Veras et~al.(2022)de~Veras, da~Silva, and da~Silva]{de_Veras_2022}
Tiago M.~L. de~Veras, Leon~D. da~Silva, and Adenilton~J. da~Silva.
\newblock Double sparse quantum state preparation.
\newblock \emph{Quantum Information Processing}, 21\penalty0 (6):\penalty0 204, 2022.
\newblock \doi{10.1007/s11128-022-03549-y}.
\newblock URL \url{https://doi.org/10.1007/s11128-022-03549-y}.

\bibitem[Dong et~al.(2021)Dong, Meng, Whaley, and Lin]{dong2021efficient}
Yulong Dong, Xiang Meng, K.~Birgitta Whaley, and Lin Lin.
\newblock Efficient phase-factor evaluation in quantum signal processing.
\newblock \emph{Phys. Rev. A}, 103:\penalty0 042419, Apr 2021.
\newblock \doi{10.1103/PhysRevA.103.042419}.
\newblock URL \url{https://link.aps.org/doi/10.1103/PhysRevA.103.042419}.

\bibitem[Eremenko and Yuditskii(2007)]{eremenko2006uniformapproximationsgnxpolynomials}
Alexandre Eremenko and Peter Yuditskii.
\newblock Uniform approximation of sgn x by polynomials and entire functions.
\newblock \emph{Journal d'Analyse Math{\'e}matique}, 101\penalty0 (1):\penalty0 313--324, 2007.
\newblock \doi{10.1007/s11854-007-0011-3}.
\newblock URL \url{https://doi.org/10.1007/s11854-007-0011-3}.

\bibitem[Feniou et~al.(2024)Feniou, Adjoua, Claudon, Zylberman, Giner, and Piquemal]{feniou2024sparse}
César Feniou, Olivier Adjoua, Baptiste Claudon, Julien Zylberman, Emmanuel Giner, and Jean-Philip Piquemal.
\newblock Sparse quantum state preparation for strongly correlated systems.
\newblock \emph{The Journal of Physical Chemistry Letters}, 15\penalty0 (11):\penalty0 3197--3205, 03 2024.
\newblock ISSN 1948-7185.
\newblock \doi{10.1021/acs.jpclett.3c03159}.
\newblock URL \url{https://doi.org/10.1021/acs.jpclett.3c03159}.

\bibitem[Feniou et~al.(2026)Feniou, Cherfan, Zylberman, Claudon, Piquemal, and Giner]{feniou2025real}
César Feniou, Christopher Cherfan, Julien Zylberman, Baptiste Claudon, Jean-Philip Piquemal, and Emmanuel Giner.
\newblock Real-space chemistry on quantum computers: a fault-tolerant algorithm with adaptive grids and a transcorrelated extension.
\newblock \emph{Digital Discovery}, 5\penalty0 (11):\penalty0 5097--5116, 11 2026.
\newblock ISSN 2635-098X.
\newblock \doi{10.1039/d6dd00002a}.
\newblock URL \url{https://doi.org/10.1039/d6dd00002a}.

\bibitem[Gidney(2019)]{gidney2019windowed}
Craig Gidney.
\newblock Windowed quantum arithmetic.
\newblock \emph{arXiv preprint arXiv:1905.07682}, 2019.

\bibitem[Gily{\'e}n et~al.(2019)Gily{\'e}n, Su, Low, and Wiebe]{qsvt}
Andr{\'a}s Gily{\'e}n, Yuan Su, Guang~Hao Low, and Nathan Wiebe.
\newblock Quantum singular value transformation and beyond: exponential improvements for quantum matrix arithmetics.
\newblock In \emph{Proceedings of the 51st Annual ACM SIGACT Symposium on Theory of Computing}, STOC 2019, page 193–204, New York, NY, USA, 2019. Association for Computing Machinery.
\newblock ISBN 9781450367059.
\newblock \doi{10.1145/3313276.3316366}.
\newblock URL \url{https://doi.org/10.1145/3313276.3316366}.

\bibitem[Gleinig and Hoefler(2021)]{quantum_dac}
Niels Gleinig and Torsten Hoefler.
\newblock An efficient algorithm for sparse quantum state preparation.
\newblock In \emph{2021 58th ACM/IEEE Design Automation Conference (DAC)}, pages 433--438, 2021.
\newblock \doi{10.1109/DAC18074.2021.9586240}.

\bibitem[Gonzalez-Conde et~al.(2024)Gonzalez-Conde, Watts, Rodriguez-Grasa, and Sanz]{gonzalez2024efficient}
Javier Gonzalez-Conde, Thomas~W. Watts, Pablo Rodriguez-Grasa, and Mikel Sanz.
\newblock Efficient quantum amplitude encoding of polynomial functions.
\newblock \emph{{Quantum}}, 8:\penalty0 1297, March 2024.
\newblock ISSN 2521-327X.
\newblock \doi{10.22331/q-2024-03-21-1297}.
\newblock URL \url{https://doi.org/10.22331/q-2024-03-21-1297}.

\bibitem[{Grover} and {Rudolph}(2002)]{grover2002creatingsuperpositionscorrespondefficiently}
Lov {Grover} and Terry {Rudolph}.
\newblock {Creating superpositions that correspond to efficiently integrable probability distributions}.
\newblock \emph{arXiv e-prints}, art. quant-ph/0208112, August 2002.
\newblock \doi{10.48550/arXiv.quant-ph/0208112}.

\bibitem[Guseynov and Liu(2026)]{guseynov2026efficient}
Nikita Guseynov and Nana Liu.
\newblock Efficient explicit circuit for quantum state preparation of piecewise continuous functions.
\newblock \emph{Phys. Rev. A}, 113:\penalty0 012604, Jan 2026.
\newblock \doi{10.1103/plc3-2jyx}.
\newblock URL \url{https://link.aps.org/doi/10.1103/plc3-2jyx}.

\bibitem[Han et~al.(2025)Han, DiBrita, Cho, Luo, and Patel]{han2025enqode}
Jason Han, Nicholas~S. DiBrita, Younghyun Cho, Hengrui Luo, and Tirthak Patel.
\newblock Enqode: Fast amplitude embedding for quantum machine learning using classical data.
\newblock In \emph{2025 62nd ACM/IEEE Design Automation Conference (DAC)}, pages 1--7, 2025.
\newblock \doi{10.1109/DAC63849.2025.11132921}.

\bibitem[H{\"a}ner et~al.(2018)H{\"a}ner, Roetteler, and Svore]{haner2018optimizing}
Thomas H{\"a}ner, Martin Roetteler, and Krysta~M Svore.
\newblock Optimizing quantum circuits for arithmetic.
\newblock \emph{arXiv preprint arXiv:1805.12445}, 2018.

\bibitem[Herbert(2021)]{herbert2021classicalsamplingcircuitquantum}
Steven Herbert.
\newblock Every classical sampling circuit is a quantum sampling circuit, 2021.
\newblock URL \url{https://arxiv.org/abs/2109.04842}.

\bibitem[Huggins et~al.(2025)Huggins, Leimkuhler, Stetina, and Whaley]{huggins2025efficient}
William~J. Huggins, Oskar Leimkuhler, Torin~F. Stetina, and K.~Birgitta Whaley.
\newblock Efficient state preparation for the quantum simulation of molecules in first quantization.
\newblock \emph{PRX Quantum}, 6:\penalty0 020319, Apr 2025.
\newblock \doi{10.1103/PRXQuantum.6.020319}.
\newblock URL \url{https://link.aps.org/doi/10.1103/PRXQuantum.6.020319}.

\bibitem[Khattar and Gidney(2025)]{Khattar2025riseofconditionally}
Tanuj Khattar and Craig Gidney.
\newblock Rise of conditionally clean ancillae for efficient quantum circuit constructions.
\newblock \emph{{Quantum}}, 9:\penalty0 1752, May 2025.
\newblock ISSN 2521-327X.
\newblock \doi{10.22331/q-2025-05-21-1752}.
\newblock URL \url{https://doi.org/10.22331/q-2025-05-21-1752}.

\bibitem[Krishnakumar et~al.(2019)Krishnakumar, Soeken, and Roetteler]{krishnakumar2019quantum}
Ranjith Krishnakumar, Mathias Soeken, and Martin Roetteler.
\newblock Efficient quantum circuit synthesis for arithmetic.
\newblock \emph{arXiv:1905.01460}, 2019.

\bibitem[Kuklinski et~al.(2025)Kuklinski, Rempfer, Obenland, and Elenewski]{kuklinski2025simpler}
Parker Kuklinski, Benjamin Rempfer, Kevin Obenland, and Justin Elenewski.
\newblock A simpler gaussian state-preparation.
\newblock 08 2025.
\newblock \doi{10.48550/arXiv.2508.03987}.

\bibitem[Lin(2026)]{lin2026quantumalgorithms}
Lin Lin.
\newblock Lecture notes on quantum algorithms for scientific computation.
\newblock \url{https://math.berkeley.edu/~linlin/qasc/live_notes_0202.pdf}, 2026.
\newblock Live notes, updated February 2, 2026.

\bibitem[Liu et~al.(2023)Liu, An, Fang, Wang, Low, and Jordan]{reac_diffu}
Jin-Peng Liu, Dong An, Di~Fang, Jiasu Wang, Guang~Hao Low, and Stephen Jordan.
\newblock {Efficient Quantum Algorithm for Nonlinear Reaction{\textendash}Diffusion Equations and Energy Estimation}.
\newblock \emph{Commun. Math. Phys.}, 404\penalty0 (2):\penalty0 963--1020, 2023.
\newblock \doi{10.1007/s00220-023-04857-9}.

\bibitem[Marin-Sanchez et~al.(2023)Marin-Sanchez, Gonzalez-Conde, and Sanz]{qalgapproxfuncload}
Gabriel Marin-Sanchez, Javier Gonzalez-Conde, and Mikel Sanz.
\newblock Quantum algorithms for approximate function loading.
\newblock \emph{Phys. Rev. Res.}, 5:\penalty0 033114, Aug 2023.
\newblock \doi{10.1103/PhysRevResearch.5.033114}.
\newblock URL \url{https://link.aps.org/doi/10.1103/PhysRevResearch.5.033114}.

\bibitem[Martyn et~al.(2021)Martyn, Rossi, Tan, and Chuang]{Martyn_2021}
John~M. Martyn, Zane~M. Rossi, Andrew~K. Tan, and Isaac~L. Chuang.
\newblock Grand unification of quantum algorithms.
\newblock \emph{PRX Quantum}, 2\penalty0 (4), December 2021.
\newblock ISSN 2691-3399.
\newblock \doi{10.1103/prxquantum.2.040203}.
\newblock URL \url{http://dx.doi.org/10.1103/PRXQuantum.2.040203}.

\bibitem[McArdle et~al.(2026)McArdle, Gily\'en, and Berta]{mcardle2022quantum}
Sam McArdle, Andr\'as Gily\'en, and Mario Berta.
\newblock Quantum state preparation without coherent arithmetic.
\newblock \emph{Phys. Rev. Lett.}, 136:\penalty0 240603, Jun 2026.
\newblock \doi{10.1103/ntvs-c48s}.
\newblock URL \url{https://link.aps.org/doi/10.1103/ntvs-c48s}.

\bibitem[Montanaro(2015)]{Montanaro_2015}
Ashley Montanaro.
\newblock Quantum speedup of monte carlo methods.
\newblock \emph{Proceedings of the Royal Society A: Mathematical, Physical and Engineering Sciences}, 471\penalty0 (2181):\penalty0 20150301, September 2015.
\newblock ISSN 1471-2946.
\newblock \doi{10.1098/rspa.2015.0301}.
\newblock URL \url{http://dx.doi.org/10.1098/rspa.2015.0301}.

\bibitem[Moosa et~al.(2023)Moosa, Watts, Chen, Sarma, and McMahon]{moosa2024linear}
Mudassir Moosa, Thomas~W Watts, Yiyou Chen, Abhijat Sarma, and Peter~L McMahon.
\newblock Linear-depth quantum circuits for loading fourier approximations of arbitrary functions.
\newblock \emph{Quantum Science and Technology}, 9\penalty0 (1):\penalty0 015002, oct 2023.
\newblock \doi{10.1088/2058-9565/acfc62}.
\newblock URL \url{https://doi.org/10.1088/2058-9565/acfc62}.

\bibitem[Motlagh and Wiebe(2024)]{gqsp}
Danial Motlagh and Nathan Wiebe.
\newblock Generalized quantum signal processing.
\newblock \emph{PRX Quantum}, 5:\penalty0 020368, Jun 2024.
\newblock \doi{10.1103/PRXQuantum.5.020368}.
\newblock URL \url{https://link.aps.org/doi/10.1103/PRXQuantum.5.020368}.

\bibitem[O'Brien and S{\"{u}}nderhauf(2025)]{o2025quantum}
Oliver O'Brien and Christoph S{\"{u}}nderhauf.
\newblock Quantum state preparation via piecewise {QSVT}.
\newblock \emph{{Quantum}}, 9:\penalty0 1786, July 2025.
\newblock ISSN 2521-327X.
\newblock \doi{10.22331/q-2025-07-03-1786}.
\newblock URL \url{https://doi.org/10.22331/q-2025-07-03-1786}.

\bibitem[Pino et~al.(2021)Pino, Dreiling, Figgatt, Gaebler, Moses, Allman, Baldwin, Foss-Feig, Hayes, Mayer, Ryan-Anderson, and Neyenhuis]{Pino_2021}
J.~M. Pino, J.~M. Dreiling, C.~Figgatt, J.~P. Gaebler, S.~A. Moses, M.~S. Allman, C.~H. Baldwin, M.~Foss-Feig, D.~Hayes, K.~Mayer, C.~Ryan-Anderson, and B.~Neyenhuis.
\newblock Demonstration of the trapped-ion quantum ccd computer architecture.
\newblock \emph{Nature}, 592\penalty0 (7853):\penalty0 209–213, April 2021.
\newblock ISSN 1476-4687.
\newblock \doi{10.1038/s41586-021-03318-4}.
\newblock URL \url{http://dx.doi.org/10.1038/s41586-021-03318-4}.

\bibitem[Piroli et~al.(2024)Piroli, Styliaris, and Cirac]{Piroli_2024}
Lorenzo Piroli, Georgios Styliaris, and J.~Ignacio Cirac.
\newblock Approximating many-body quantum states with quantum circuits and measurements.
\newblock \emph{Physical Review Letters}, 133\penalty0 (23), December 2024.
\newblock ISSN 1079-7114.
\newblock \doi{10.1103/physrevlett.133.230401}.
\newblock URL \url{http://dx.doi.org/10.1103/PhysRevLett.133.230401}.

\bibitem[Plesch and Brukner(2011)]{plesch2011quantum}
Martin Plesch and {\ifmmode \check{C}\else \v{C}\fi{}}aslav Brukner.
\newblock Quantum-state preparation with universal gate decompositions.
\newblock \emph{Phys. Rev. A}, 83:\penalty0 032302, Mar 2011.
\newblock \doi{10.1103/PhysRevA.83.032302}.
\newblock URL \url{https://link.aps.org/doi/10.1103/PhysRevA.83.032302}.

\bibitem[Ransford et~al.(2026)Ransford, Allman, Arkinstall, Campora, Cooper, Delaney, Dreiling, Estey, Figgatt, Hall, Husain, Isanaka, Kennedy, Kotibhaskar, Madjarov, Mayer, Milne, Park, Reed, Ancona, Andersen, Andres-Martinez, Angenent, Argueta, Arkin, Ascarrunz, Baker, Barnes, Bartolotta, Berg, Besand, Bjork, Blain, Blanchard, Blume-Kohout, Bohn, Borgna, Botamanenko, Boutelle, Brown, Buckingham, Burdick, Burton, Carey, Carron, Chambers, Chan, Children, Colussi, Crepinsek, Cureton, Davies, Davis, DeCross, Deen, Delaney, DelVento, DeSalvo, Dominy, Drotar, Duncan, Eccles, Edgington, Erickson, Erickson, Ertsgaard, Esposito, Evans, Evans, Fabrikant, Fischer, Foltz, Foss-Feig, Francois, Freyberg, Gao, Garay, Garvin, Gaudiosi, Gilbreth, Giles, Glynn, Graves, Hansen, Hayes, Heidemann, Higashi, Hilbun, Hines, Hlavaty, Hoffman, Hoffman, Holliman, Hooper, Horning, Hostetter, Hothem, Houlton, Hout, Hutson, Jacobs, Jacobs, Johannsen, Johansen, Jones, Julian, Jung, Keay, Klein, Koch, Kondo, Kong, Kosto, Lawrence, Liefer,
  Lollie, Lucchetti, Lysne, Lytle, MacPherson, Malm, Mather, Mathewson, Maxwell, McCaffrey, McDougall, Mendoza, Miller, Mills, Morrison, Narmour, Nguyen, Nugent, Olson, Ouellette, Parks, Peters, Peterson, Petricka, Pino, Polito, Potter, Preidl, Price, Proctor, Pugh, Ratcliff, Raymondson, Rhodes, Roman, Roy, Ryan-Anderson, Sanchez, Sangiolo, Sawadski, Schaffer, Schow, Sedlacek, Semenenko, Shevchuk, Shore, Siegfried, Singhal, Sivarajah, Skripka, Sletten, Spaun, Sprenkle, Stoufer, Tader, Taylor, Thompson, Tobey, Tran, Tran, Vittorini, Volin, Walker, White, Williams, Wilson, Wolf, Wringe, Young, Zheng, Zuraski, Baldwin, Chernoguzov, Gaebler, Sanders, Neyenhuis, Stutz, and Bohnet]{helios}
Anthony Ransford, M.~S. Allman, Jake Arkinstall, J.~P. Campora, Samuel~F. Cooper, Robert~D. Delaney, Joan~M. Dreiling, Brian Estey, Caroline Figgatt, Alex Hall, Ali~A. Husain, Akhil Isanaka, Colin~J. Kennedy, Nikhil Kotibhaskar, Ivaylo~S. Madjarov, Karl Mayer, Alistair~R. Milne, Annie~J. Park, Adam~P. Reed, Riley Ancona, Molly~P. Andersen, Pablo Andres-Martinez, Will Angenent, Liz Argueta, Benjamin Arkin, Leonardo Ascarrunz, William Baker, Corey Barnes, John Bartolotta, Jordan Berg, Ryan Besand, Bryce Bjork, Matt Blain, Paul Blanchard, Robin Blume-Kohout, Matt Bohn, Agust{\'\i}n Borgna, Daniel~Y. Botamanenko, Robert Boutelle, Natalie Brown, Grant~T. Buckingham, Nathaniel~Q. Burdick, William~Cody Burton, Varis Carey, Christopher~J. Carron, Joe Chambers, Jia~Wen Chan, John Children, Victor~E. Colussi, Steven Crepinsek, Andrew Cureton, Joe Davies, Daniel Davis, Matthew DeCross, David Deen, Conor Delaney, Davide DelVento, B.~J. DeSalvo, Jason Dominy, Sydney Drotar, Ross Duncan, Vanya Eccles, Alec Edgington, Neal
  Erickson, Stephen Erickson, Christopher~T. Ertsgaard, Jay Esposito, Bruce Evans, Tyler Evans, Maya~I. Fabrikant, Andrew Fischer, Cameron Foltz, Michael Foss-Feig, David Francois, Brad Freyberg, Charles Gao, Robert Garay, Jane Garvin, David~M. Gaudiosi, Christopher~N. Gilbreth, Josh Giles, Erin Glynn, Jeff Graves, Azure Hansen, David Hayes, Lukas Heidemann, Bob Higashi, Tyler Hilbun, Jordan Hines, Ariana Hlavaty, Kyle Hoffman, Ian~M. Hoffman, Craig Holliman, Isobel Hooper, Bob Horning, James Hostetter, Daniel Hothem, Jack Houlton, Jared Hout, Ross Hutson, Ryan~T. Jacobs, Trent Jacobs, Melf Johannsen, Jacob Johansen, Loren Jones, Sydney Julian, Ryan Jung, Aidan Keay, Todd Klein, Mark Koch, Ryo Kondo, Chang Kong, Asa Kosto, Alan Lawrence, David Liefer, Michelle Lollie, Dominic Lucchetti, Nathan~K. Lysne, Christian Lytle, Callum MacPherson, Andrew Malm, Spencer Mather, Brian Mathewson, Daniel Maxwell, Lauren McCaffrey, Hannah McDougall, Robin Mendoza, David~B. Miller, Michael Mills, Richard Morrison, Louis
  Narmour, Nhung Nguyen, Lora Nugent, Scott Olson, Daniel Ouellette, Jeremy Parks, Zach Peters, Timothy~A. Peterson, Jessie Petricka, Juan~M. Pino, Frank Polito, Andrew~C. Potter, Matthias Preidl, Gabriel Price, Timothy Proctor, McKinley Pugh, Noah Ratcliff, Daisy Raymondson, Peter Rhodes, Conrad Roman, Craig Roy, Ciaran Ryan-Anderson, Fernando~Betanzo Sanchez, George Sangiolo, Tatiana Sawadski, Andrew Schaffer, Peter Schow, Jon Sedlacek, Henry Semenenko, Peter Shevchuk, Susan Shore, Peter Siegfried, Kartik Singhal, Seyon Sivarajah, Thomas Skripka, Lucas Sletten, Ben Spaun, R.~Tucker Sprenkle, Paul Stoufer, Mariel Tader, Stephen~F. Taylor, Travis~H. Thompson, Raanan Tobey, Anh Tran, Tam Tran, Grahame Vittorini, Curtis Volin, Jim Walker, Sam White, Garrett~R. Williams, Douglas Wilson, Quinn Wolf, Chester Wringe, Kevin Young, Jian Zheng, Kristen Zuraski, Charles~H. Baldwin, Alex Chernoguzov, John~P. Gaebler, Steven~J. Sanders, Brian Neyenhuis, Russell Stutz, and Justin~G. Bohnet.
\newblock A 98-qubit trapped-ion quantum computer with all-to-all connectivity.
\newblock \emph{Nature}, 655\penalty0 (8121):\penalty0 81--86, 2026.
\newblock \doi{10.1038/s41586-026-10676-4}.
\newblock URL \url{https://doi.org/10.1038/s41586-026-10676-4}.

\bibitem[Rattew and Koczor(2022)]{rattew2022preparing}
Arthur~G Rattew and B{\'a}lint Koczor.
\newblock Preparing arbitrary continuous functions in quantum registers with logarithmic complexity.
\newblock \emph{arXiv preprint arXiv:2205.00519}, 2022.

\bibitem[Rattew and Rebentrost(2023)]{rattew2023nonlineartransformationsquantumamplitudes}
Arthur~G. Rattew and Patrick Rebentrost.
\newblock Non-linear transformations of quantum amplitudes: Exponential improvement, generalization, and applications.
\newblock 2023.
\newblock URL \url{https://arxiv.org/abs/2309.09839}.

\bibitem[Rethinasamy et~al.(2024)Rethinasamy, LaBorde, and Wilde]{logdepth_hamming}
Soorya Rethinasamy, Margarite~L. LaBorde, and Mark~M. Wilde.
\newblock Logarithmic-depth quantum circuits for hamming weight projections.
\newblock \emph{Phys. Rev. A}, 110:\penalty0 052401, Nov 2024.
\newblock \doi{10.1103/PhysRevA.110.052401}.
\newblock URL \url{https://link.aps.org/doi/10.1103/PhysRevA.110.052401}.

\bibitem[Rosenkranz et~al.(2025)Rosenkranz, Brunner, Marin-Sanchez, Fitzpatrick, Dilkes, Tang, Kikuchi, and Benedetti]{rosenkranz2025quantum}
Matthias Rosenkranz, Eric Brunner, Gabriel Marin-Sanchez, Nathan Fitzpatrick, Silas Dilkes, Yao Tang, Yuta Kikuchi, and Marcello Benedetti.
\newblock Quantum state preparation for multivariate functions.
\newblock \emph{{Quantum}}, 9:\penalty0 1703, April 2025.
\newblock ISSN 2521-327X.
\newblock \doi{10.22331/q-2025-04-09-1703}.
\newblock URL \url{https://doi.org/10.22331/q-2025-04-09-1703}.

\bibitem[Saibaba(2021)]{saibaba2021approximatingmonomialsusingchebyshev}
Arvind Saibaba.
\newblock Approximating monomials using chebyshev polynomials.
\newblock 01 2021.
\newblock \doi{10.48550/arXiv.2101.06818}.

\bibitem[Sanders et~al.(2019)Sanders, Low, Scherer, and Berry]{sanders2019black}
Yuval~R. Sanders, Guang~Hao Low, Artur Scherer, and Dominic~W. Berry.
\newblock Black-box quantum state preparation without arithmetic.
\newblock \emph{Physical Review Letters}, 122:\penalty0 020502, 2019.
\newblock \doi{10.1103/PhysRevLett.122.020502}.

\bibitem[S{\"{u}}nderhauf et~al.(2024)S{\"{u}}nderhauf, Campbell, and Camps]{sunderhauf2024block}
Christoph S{\"{u}}nderhauf, Earl Campbell, and Joan Camps.
\newblock Block-encoding structured matrices for data input in quantum computing.
\newblock \emph{{Quantum}}, 8:\penalty0 1226, January 2024.
\newblock ISSN 2521-327X.
\newblock \doi{10.22331/q-2024-01-11-1226}.
\newblock URL \url{https://doi.org/10.22331/q-2024-01-11-1226}.

\bibitem[Sünderhauf(2023)]{gqsvt}
Christoph Sünderhauf.
\newblock Generalized quantum singular value transformation, 2023.
\newblock URL \url{https://arxiv.org/abs/2312.00723}.

\bibitem[Tubman et~al.(2018)]{tubman2018postponing}
Norm~M. Tubman et~al.
\newblock Postponing the orthogonality catastrophe: efficient state preparation for electronic structure simulations.
\newblock \emph{arXiv:1809.05523}, 2018.

\bibitem[Vartiainen et~al.(2004)Vartiainen, M{\"o}tt{\"o}nen, and Salomaa]{vartiainen2004efficient}
Juha~J. Vartiainen, Mikko M{\"o}tt{\"o}nen, and Martti~M. Salomaa.
\newblock Efficient decomposition of quantum gates.
\newblock \emph{Phys. Rev. Lett.}, 92:\penalty0 177902, Apr 2004.
\newblock \doi{10.1103/PhysRevLett.92.177902}.
\newblock URL \url{https://link.aps.org/doi/10.1103/PhysRevLett.92.177902}.

\bibitem[Vedral et~al.(1996)Vedral, Barenco, and Ekert]{vedral1996quantum}
Vlatko Vedral, Adriano Barenco, and Artur Ekert.
\newblock Quantum networks for elementary arithmetic operations.
\newblock \emph{Phys. Rev. A}, 54:\penalty0 147--153, Jul 1996.
\newblock \doi{10.1103/PhysRevA.54.147}.
\newblock URL \url{https://link.aps.org/doi/10.1103/PhysRevA.54.147}.

\bibitem[Wang et~al.(2022)Wang, Wang, He, Cui, Shi, Shang, Li, Li, Li, Wei, et~al.]{wang2022quantum}
Sheng Wang, Zhiqiang Wang, Rui He, Guoping Cui, Shijie Shi, Rong Shang, Jing Li, Yi~Li, Wei Li, Zhaohui Wei, et~al.
\newblock Quantum state preparation with optimal circuit depth.
\newblock \emph{New Journal of Physics}, 24:\penalty0 103004, 2022.
\newblock \doi{10.1088/1367-2630/ac95fc}.

\bibitem[Wang et~al.(2021)Wang, Wang, Cui, Shi, Shang, Fan, Li, Wei, and Gu]{bausch2022fast}
Shengbin Wang, Zhimin Wang, Guolong Cui, Shangshang Shi, Ruimin Shang, Lixin Fan, Wendong Li, Zhiqiang Wei, and Yongjian Gu.
\newblock Fast black-box quantum state preparation based on linear combination of unitaries.
\newblock \emph{Quantum Information Processing}, 20\penalty0 (8):\penalty0 270, 2021.
\newblock \doi{https://doi.org/10.1007/s11128-021-03203-z}.
\newblock URL \url{https://doi.org/10.1007/s11128-021-03203-z}.

\bibitem[Ying(2022)]{ying2022stable}
Lexing Ying.
\newblock Stable factorization for phase factors of quantum signal processing.
\newblock \emph{{Quantum}}, 6:\penalty0 842, October 2022.
\newblock ISSN 2521-327X.
\newblock \doi{10.22331/q-2022-10-20-842}.
\newblock URL \url{https://doi.org/10.22331/q-2022-10-20-842}.

\bibitem[Zhao et~al.(2021)Zhao, Fitzsimons, Rebentrost, Dunjko, and Fitzsimons]{zhao2021smooth}
Zhikuan Zhao, Jack~K. Fitzsimons, Patrick Rebentrost, Vedran Dunjko, and Joseph~F. Fitzsimons.
\newblock Smooth input preparation for quantum and quantum-inspired machine learning.
\newblock \emph{Quantum Machine Intelligence}, 3\penalty0 (1):\penalty0 14, 2021.
\newblock \doi{10.1007/s42484-021-00045-x}.
\newblock URL \url{https://doi.org/10.1007/s42484-021-00045-x}.

\bibitem[{Zi} et~al.(2025){Zi}, {Nie}, and {Sun}]{zi2024shallowquantumcircuitimplementation}
Wei {Zi}, Junhong {Nie}, and Xiaoming {Sun}.
\newblock {Shallow Quantum Circuit Implementation of Symmetric Functions With Limited Ancillary Qubits}.
\newblock \emph{IEEE Transactions on Computer Aided Design}, 44\penalty0 (8):\penalty0 3060--3072, January 2025.
\newblock \doi{10.1109/TCAD.2025.3539002}.

\bibitem[Zylberman(2024)]{zylberman2024fast}
Julien Zylberman.
\newblock Fast laplace transforms on quantum computers.
\newblock 12 2024.
\newblock \doi{10.48550/arXiv.2412.05173}.

\bibitem[Zylberman and Debbasch(2024)]{zylberman2024efficient}
Julien Zylberman and Fabrice Debbasch.
\newblock Efficient quantum state preparation with walsh series.
\newblock \emph{Phys. Rev. A}, 109:\penalty0 042401, Apr 2024.
\newblock \doi{10.1103/PhysRevA.109.042401}.
\newblock URL \url{https://link.aps.org/doi/10.1103/PhysRevA.109.042401}.

\bibitem[Zylberman et~al.(2022)Zylberman, Di~Molfetta, Brachet, Loureiro, and Debbasch]{madelung}
Julien Zylberman, Giuseppe Di~Molfetta, Marc Brachet, Nuno~F. Loureiro, and Fabrice Debbasch.
\newblock Quantum simulations of hydrodynamics via the madelung transformation.
\newblock \emph{Phys. Rev. A}, 106:\penalty0 032408, Sep 2022.
\newblock \doi{10.1103/PhysRevA.106.032408}.
\newblock URL \url{https://link.aps.org/doi/10.1103/PhysRevA.106.032408}.

\bibitem[Zylberman et~al.(2025)Zylberman, Nzongani, Simonetto, and Debbasch]{zylberman2025efficient}
Julien Zylberman, Ugo Nzongani, Andrea Simonetto, and Fabrice Debbasch.
\newblock Efficient quantum circuits for non-unitary and unitary diagonal operators with space-time-accuracy trade-offs.
\newblock \emph{ACM Transactions on Quantum Computing}, 6\penalty0 (2), April 2025.
\newblock \doi{10.1145/3718348}.
\newblock URL \url{https://doi.org/10.1145/3718348}.

\bibitem[Zylberman et~al.(2026)Zylberman, Fredon, Loureiro, and Debbasch]{zylberman2025trotter}
Julien Zylberman, Thibault Fredon, Nuno~F Loureiro, and Fabrice Debbasch.
\newblock Trotter-based quantum algorithm for solving transport equations with exponentially fewer time-steps.
\newblock \emph{Quantum Science and Technology}, 11\penalty0 (2):\penalty0 025015, mar 2026.
\newblock \doi{10.1088/2058-9565/ae488e}.
\newblock URL \url{https://doi.org/10.1088/2058-9565/ae488e}.

\end{thebibliography}

\onecolumngrid
\appendix

\newpage
\appendix
\section{Quantum Signal Processing}\label{app:quantum_signal_processing}

In this section, we recall quantum signal processing theorems. We use the notion of Project Unitary Encoding (PUE) which generalizes the concept of block-encoding: one recovers the target operation by projecting a unitary onto certain subspaces while block-encodings embed an operator as a top-left block of a larger unitary. Therefore, PUE gives more flexibility in how the unitary is constructed.

\begin{Def}[(Symmetric) Projected Unitary Encoding] \index{Symmetric Projected Unitary Encoding}\index{Projected Unitary Encoding}
Let $U$ be a unitary operator and $\square_L, \square_R$ be partial isometries. $(U, \square_L, \square_R)$ is said to be a Projected Unitary Encoding (PUE) of $A=\square_L^\dag U\square_R$. If $\square_L=\square_R=\square$ and $U$ is also symmetric, $(U, \square)$ is called a Symmetric Projected Unitary Encoding (SPUE) of $A=\square^\dag U\square$.
\end{Def}

Theorem~\ref{the:gqsp} details single-ancilla constructions to encode polynomials of unitaries~\cite{gqsp}.

\begin{The}{(Generalized Quantum Signal Processing)}
Let $\Upsilon$ be a degree $d\in \mathbb N$ complex polynomial and $U$ be a unitary.
If $\|\Upsilon(V)\|\leq 1$ for all unitaries $V$, then we can construct a PUE $(W, \ket0, \ket0)$ of $\Upsilon(U)$
using $d$ controlled-$U$ operations, $\mathcal O(d)$ additional single-qubit gates and $1$ ancilla qubit.
\label{the:gqsp}
\end{The}

Theorem~\ref{the:gqsp} can also be used to encode polynomials of symmetric matrices. The procedure uses the notion of signal processing polynomial, given in Definition~\ref{def:signal_processing_polynomial}.

\begin{Def}[Signal processing polynomial]
Consider a degree-$d$ complex polynomial $\upsilon(x)=\sum_{n=0}^da_nT_n(x),$ written in the basis of Chebyshev polynomials. Its associated signal processing polynomial is defined to be $\Upsilon(z)=\sum_{n=0}^da_nz^n$. 
\label{def:signal_processing_polynomial}
\end{Def}

The complexity of implementing polynomial ultimately depends on its scaling factor, defined as follows.

\begin{Def}[Scaling factor]\index{Scaling factor}
The scaling factor $\beta$ of $\upsilon$ is defined by:
\begin{equation}
\beta = \frac{\max_{z\in \partial B_1(0)}\left|\Upsilon(z)\right|}{\max_{x\in [-1, 1]}|\upsilon(x)|},
\end{equation}
where $\Upsilon$ is the signal processing polynomial of $\upsilon$.
\label{def:scalingfac}
\end{Def}

Theorem~\ref{the:gqet} is derived from Theorem~\ref{the:gqsp} and used in our construction~\cite{gqsvt}.

\begin{The}{(Generalized Quantum Eigenvalue Transform)}
Let $(U, \square)$ be a SPUE of $A$.
Consider a degree-$d$ complex polynomial $\upsilon(x)=\sum_{n=0}^d a_n T_n(x)$ written in the basis of Chebyshev polynomials.
Consider its associated signal processing polynomial $\Upsilon(z)=\sum_{n=0}^d a_n z^n$.
Assume that $\max_{z\in\partial B_1(0)}|\Upsilon(z)|\leq1$, namely that its scaling factor satisfies $\beta\leq 1$.
Let $\mathcal G$ be the GQSP (Theorem~\ref{the:gqsp}) transformation of the qubitized walk operator of $(U, \square)$ applying the polynomial $\Upsilon$.
Then, $\left(\mathcal G, \ket0\otimes \square\right)$ is a PUE of $\upsilon(A)$.
\label{the:gqet}
\end{The}

If a polynomial $\upsilon$ has a scaling factor $\beta>1$, the polynomial $\upsilon/\beta$ can be implemented. Moreover, as shown in~\cite{gqsvt}, $\beta$ grows at most logarithmically with the degree of $\upsilon$.

\section{Proofs}\label{app:proofs}
In the following, technical details and formal proofs of the formal version of the theorems are presented.

\subsection{Objective}
Let us begin by defining the operator for which we aim to construct a symmetric projected unitary encoding.

\begin{Def}[Linear diagonal operator]\index{Linear diagonal operator}
Let $n\geq1$. Define
\begin{equation}
L_n=\frac1{1-1/2^n}\sum_{x\in B_n}x\ket x\bra x.
\end{equation}
\end{Def}

The key observation is the following.

\begin{Prop} For every integer $n\geq1$, the following identity holds true:
\begin{equation}
\frac1{1-1/2^n}\sum_{k=1}^n\frac{Z_k}{2^k}=1-2L_n.
\end{equation}
\label{prop:pauli_decomposition}
\end{Prop}

\begin{proof}
For every $a\in \{0, 1\}$, 
\begin{equation}
(-1)^a=1-2a.
\end{equation}
Therefore,
\begin{equation}
\begin{split}
\forall x\in B_n:\frac1{1-1/2^n}\sum_{k=1}^n\frac{Z_k}{2^k}\ket x&=\frac1{1-1/2^n}\sum_{k=1}^n\frac{(-1)^{x_k}}{2^k}\ket x\\
&=\frac1{1-1/2^n}\sum_{k=1}^n\frac{1-2x_k}{2^k}\ket x\\
&=\ket x-\frac2{1-1/2^n}\sum_{k=1}^n\frac{x_k}{2^k}\ket x\\
&=\ket x-\frac2{1-1/2^n}x\ket x\\
&=(1-2L_n)\ket x.
\end{split}
\end{equation}
Because the identity holds true for every computational basis state in $B_n$, the identity holds true in general.
\end{proof}

Our strategy to construct $L_n$ is to design a low-depth implementation of $\frac1{1-1/2^n}\sum_{k=1}^n\frac{Z_k}{2^k}$. To do so, we adopt a PREPARE/SELECT strategy where the coefficients $1/2^k$ needs to be loaded in a quantum state via the following quantum state: 
\begin{equation}
\ket{\beta_n}=\frac1{\sqrt{1-1/2^n}}\sum_{k=1}^n\frac1{\sqrt{2^k}}\ket{1/2^k},
\end{equation}
so that, one can produce the $\sum Z_k/2^k$ using a unary SELECT operator. Thus, one can show
\begin{equation}
\left(\sum_{k=1}^n\ket{1/2^k}\bra{1/2^k}\otimes Z_k+\left(1-\sum_{k=1}^n\ket{1/2^k}\bra{1/2^k}\right)\otimes 1, \ket{\beta_n}\right)
\end{equation}
is a SPUE of $\frac1{1-1/2^n}\sum_{k=1}^n\frac{Z_k}{2^k}$.

\subsection{PREPARE operator}
\subsubsection{First layer of quantum gates}
In order to prepare $\ket{\beta_n}$, we note that we can prepare a state $\ket{\alpha_n}$ with constant overlap with $\ket{\beta_n}$ at constant cost, before projecting onto $\ket{\beta_n}$ with logarithmic depth.

\begin{Def}[$\ket{\alpha_n}$]\index{$\ket{\alpha_n}$} Define the $n$-qubit state $\ket{\alpha_n}$ by:
\begin{equation}
\ket{\alpha_n}=\bigotimes_{k=1}^n\left(\sqrt{1-\frac1{2^k+1}}\ket0+\sqrt{\frac1{2^k+1}}\ket1\right).
\end{equation}
\end{Def}

\begin{Prop} We can construct a constant depth unitary $A_n$ such that $A_n\ket 0=\ket{\alpha_n}$.
\end{Prop}

\begin{proof}
Take
\begin{equation}
A_n=\bigotimes_{k=1}^n\begin{pmatrix}
\sqrt{1-\frac1{2^k+1}}  & -\sqrt{\frac1{2:^k+1}}\\
\sqrt{\frac1{2^k+1}}    & \sqrt{1-\frac1{2^k+1}}
\end{pmatrix}\\
=\bigotimes_{k=1}^n\exp\left(-i\left(\sin^{-1}\left(\sqrt{\frac1{2^k+1}}\right)\right)Y\right).
\end{equation}
\end{proof}

Now, let us formalize the definition of $\ket{\beta_n}$.

\begin{Def}[$\ket{\beta_n}$]\index{$\ket{\beta_n}$} Define the $n$-qubit state $\ket{\beta_n}$ by:
\begin{equation}
\ket{\beta_n}=\frac1{\sqrt{1-\frac1{2^n}}}\sum_{k=1}^{n}\sqrt{\frac{1}{2^k}}\ket{1/2^k}.
\end{equation}
\end{Def}

$\ket{\beta_n}$ and $\ket{\alpha_n}$ are related by the following proposition.

\begin{Prop} The projection of $\ket{\alpha_n}$ on the span of $\{\ket{1/2^k}\}_{k=1}^n$ is $\ket{\beta_n}$.
\label{prop:projection}
\end{Prop}

\begin{proof} Note that for each $1\leq k \leq n$,
\begin{equation}
\braket{1/2^k|\alpha_n}=\sqrt{\frac{\frac{1}{2^k+1}}{1-\frac1{2^k+1}}}\prod_{j=1}^n\sqrt{1-\frac1{2^j+1}}=\frac1{\sqrt{2^k}}\prod_{j=1}^n\sqrt{1-\frac1{2^j+1}}.
\end{equation}
It follows that
\begin{equation}
\sum_{k=1}^n\ket{1/2^k}\braket{1/2^k|\alpha_n}=\left(\prod_{j=1}^n\sqrt{1-\frac1{2^j+1}}\right)\sqrt{1-1/2^n}\ket{\beta_n}.
\end{equation} 
\end{proof}

As a consequence, the following definition is well-posed.

\begin{Def}[$p_n$]\index{$p_n$} Define $p_n\in [0, 1]$ such that:
\begin{equation}
\ket{\alpha_n}=\sqrt{p_n}\ket{\beta_n}+\sqrt{1-p_n}\ket{\gamma_n},
\end{equation}
for some state $\ket{\gamma_n}$ that is orthogonal to $\ket{\beta_n}$.
\end{Def}

Moreover, we can lower bound the overlap between $\ket{\alpha_n}$ and $\ket{\beta_n}$ uniformly in $n\in \mathbb N$.

\begin{Prop} For every integer $n\geq1$, 
\begin{equation}
p_n=\sum_{k=1}^n|\braket{1/2^k|\alpha_n}|^2\geq \frac16.
\label{eq:pn}
\end{equation}
\end{Prop}

\begin{proof} Note that $p_n=|\braket{\beta_n|\alpha_n}|^2$. According to Proposition \ref{prop:projection}, $p_n$ is the square of the norm of the orthogonal projection of $\ket{\alpha_n}$ on the subspace spanned by $\{\ket{1/2^k}\}_{k=1}^n$. We can compute this squared norm by summing the square overlap with each of the generating vector. If $1\leq k\leq n$, then
\begin{equation}
|\braket{1/2^k|\alpha_n}|^2=\frac1{2^k+1}\prod_{j=1, j\neq k}^n\left(1-\frac1{2^j+1}\right)\geq \frac1{2^{k+1}}\prod_{j=1, j\neq k}^n\left(1-\frac1{2^j+1}\right).
\end{equation}
Therefore,
\begin{equation}
\log\left(|\braket{1/2^k|\alpha_n}|^2\right)\geq -(k+1)+\sum_{j=1, j\neq k}^n\log\left(1-\frac1{2^j+1}\right).
\end{equation}
By concavity of the function $t\mapsto \log(1-t)$ on the interval $[0, 2/3]$,
\begin{equation}
\begin{split}
\log\left(|\braket{1/2^k|\alpha_n}|^2\right)&\geq -(k+1)+\sum_{j=1, j\neq k}^n\log(2/3)\frac1{2^j+1}\\
&\geq-(k+1)+\log(2/3)(1-1/2^n)-\log(2/3)/(2^{k}+1)\\
&\geq-(k+1)-\log(3/2)\\
&=-k-\log(3).
\end{split}
\end{equation}
Taking the exponential on both sides:
\begin{equation}
|\braket{1/2^k|\alpha_n}|^2\geq \frac{2^{-k}}{3}.
\end{equation}
Summing over all $k$:
\begin{equation}
p_n\geq\frac{1-1/2^n}3\geq \frac16.
\end{equation}
\end{proof}

\subsubsection{EXACT-one in logarithmic depth}
We now present a modified version of $\ket{\beta_n}$. We add a so-called flag ancilla qubits that is active only in the subspace spanned by $\{\ket{1/2^k}\}_{k=1}^n$.

\begin{Def}[$\ket{\psi_n}$]\index{$\ket{\psi_n}$} Define the $n$-qubit state $\ket{\psi_n}$ by:
\begin{equation}
\ket{\psi_n}=\sqrt{p_n}\ket{\beta_n}\ket 1+\sqrt{1-p_n}\ket{\gamma_n}\ket0.
\end{equation}
\end{Def}

Whether the so-called flag ancilla qubit should be in state $0$ or $1$ can be computed classically as follows.

\begin{Prop} Let $n\in\mathbb N$, $n = 2m + r$ where $m\in \mathbb N, r\in \{0,1\}$. Define the function $g_n:\{0, 1\}^n\to\{0, 1,2\}$ by:
\begin{equation}
g_n(x)=\left\{\begin{aligned}
&0&\text {if $\sum_{i=1}^nx_i=0$,}\\
&1&\text{if $\sum_{i=1}^nx_i=1$,}\\
&2&\text{else.}
\end{aligned}\right.
\end{equation}
We can construct a circuit $G_n$ such that $G_n\ket{x}\ket{0}^{\otimes 2}=\ket{x}\ket{g_n(x)}$ with $\mathcal O(\log(n))$ depth, $\mathcal O(n)$ size using $2(n-2)+r$ ancillae.
\end{Prop}

\begin{proof}
We define another function $h:\{0, 1,2\}^2\to\{0, 1,2\}$ by:

\begin{equation}
h(a,b)=\left\{\begin{aligned}
&0&\text {if $(a,b)=(0,0)$,}\\
&1&\text{if $(a,b)=(0,1)$ or $(a,b)=(1,0)$,}\\
&2&\text{else.}
\end{aligned}\right.
\end{equation}

For every bitstring $x$ of length $n$, write $x=x^{(l)}x^{(k)}$ the concatenation of a bitstring $x^{(l)}$ of length $l$ and a bitstring $x^{(k)}$ of length $k$ such that $k+l=n$.\\ \\

Then, note that:
\begin{equation}
g_n(x)=h\left(g_{l}\left(x^{(l)}\right),g_{k}\left(x^{(k)}\right)\right).
\label{eq:eq_4}
\end{equation}

A first step is to compute $g_n$ recursively with NOT, CNOT and Toffoli gates and with $\mathcal O(n)$ ancillae. With the following binary encryptions $0 = (0,0)$, $1=(1,0)$ and $2=(0,1)$, we can encode outputs on pairs of qubits. $g_1$ is obtained by adding one zeroed ancilla to create a two qubits circuit storing directly the output (see Figure~\ref{fig:A_1}). $g_2$ can be implemented using two zeroed ancillae, and using these ancillae to store the output (see Figure~\ref{fig:A_2}). \\ \\

\begin{figure}[H]
    \centering
    \begin{subfigure}[t]{0.4\textwidth}
    \centering
        \includegraphics[width=0.5\textwidth]{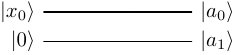}
        \caption{Quantum circuit for $g_1$: since the input is a single qubit, adding an extra zero-initialized ancilla yields the full binary encoding of the state.  \\
        }
        \label{fig:A_1}
    \end{subfigure}
    \hspace{5em}
    \begin{subfigure}[t]{0.4\textwidth}
    \centering
        \includegraphics[width=0.5\textwidth]{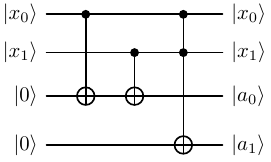}
        \caption{Quantum circuit for $g_2$: The first two CNOT gates detect whether the qubit pair has Hamming weight one and conditionally set the first ancilla qubit. The third Toffoli gate sets the second ancilla if the pair has Hamming weight two.}
        \label{fig:A_2}
    \end{subfigure}
    \caption{}
    \label{fig:A_22}
\end{figure}

Considering a bitstring $x$, we can split it in a collection of pairs and compute $g_2$ on all this pairs in parallel. If the length of $x$ is odd, we will use $g_1$ for the last qubit. This is the encoding part : we now have the information if there is zero one, exactly one one, or two ones in each pair of qubits.\\ \\
Then, we can compute $g_4(yz) = h(g_2(y),g_2(z))$ and store the result on two zeroed qubits, where $yz$ is the concatenation of $y$ and $z$ two-bits strings (see Figure~\ref{fig:A_3}). This is the transmission part. Each merging requires two new zeroed ancillae. \\

\begin{figure}[H]
    \centering
    \begin{subfigure}[b]{0.45\textwidth}
    \centering
        \hspace*{10mm}\includegraphics[width=\textwidth]{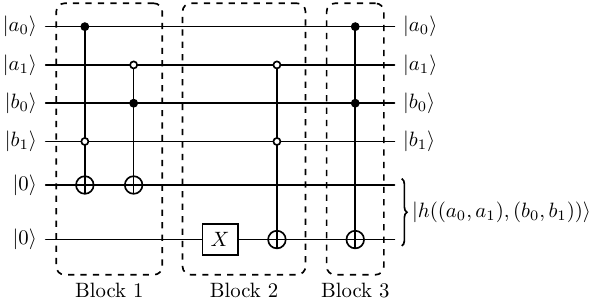}
        \caption{}
        \label{fig:A_31}
    \end{subfigure}
    \hspace{1em}
    \begin{subfigure}[b]{0.45\textwidth}
    \centering
        \hspace*{7mm}\includegraphics[width=\textwidth]{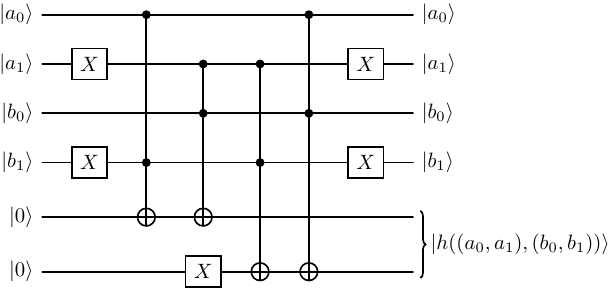}
        \caption{}
        \label{fig:A_32}
    \end{subfigure}
    \caption{The quantum circuit for $h$ can be divided into three blocks (Figure~\ref{fig:A_31}).
    The first block acts on the first ancilla. When checking whether the state is $(1,0)$, we only need to control the first qubit, since the other accessible states are $(0,0)$ and $(0,1)$. Using two Toffoli gates, it sets the ancilla to one if exactly one of the two pairs is in the state $(1,0)$. If both pairs are $(1,0)$, the Toffoli gates cancel each other out, leaving the ancilla in the zero state. \\
    The second and third blocks act on the second ancilla, setting it to one if the concatenation of the two pieces of information indicates a Hamming weight greater than two. This can be implemented using only one X gate and two Toffoli gates. The second block sets the ancilla to one if at least one of the inputs is $(0,1)$. In this case, the third block will not be activated, since the state $(1,1)$ is not accessible. If the second block is not active, the third block will activate and set the ancilla to one when both pairs are $(1,0)$. Merging all blocks, we get a 5-depth circuit (Figure~\ref{fig:A_32}).}
    \label{fig:A_3}
\end{figure}

Thus using formula~\eqref{eq:eq_4}, for every bitstring $x$ of length $n = \sum_{i=0}^{r}{a_{i}2^i}$, we can construct a circuit $G_n$ with a ladder of $h$ circuits applied in parallel that merge the informations from smaller bitstrings to get $g_n(x)$, with depth $\mathcal O(\log(n))$ and $2(n-2) + a_0$ zeroed ancillae (see example in Figure~\ref{fig:A_4} for a 7 qubits bitstring) as the encoding part requires $n$ zeroed ancillae and the transmission part requires $\sum_{i=2}^{r}{a_{i}(\sum_{j=0}^{i-1}{2^j})}+2(\sum_{i=0}^{r}{a_{i}}-1) -2 = \sum_{i=2}^{r}{a_{i}(2^i}-2)+2\sum_{i=0}^{r}{a_{i}} - 4 = n-4+a_0 $ zeroed ancillae. The ancillae can be uncomputed at the expense of doubling the circuit depth.\\ \\

\begin{figure}[H]
    \centering
    \includegraphics[width=0.6\linewidth]{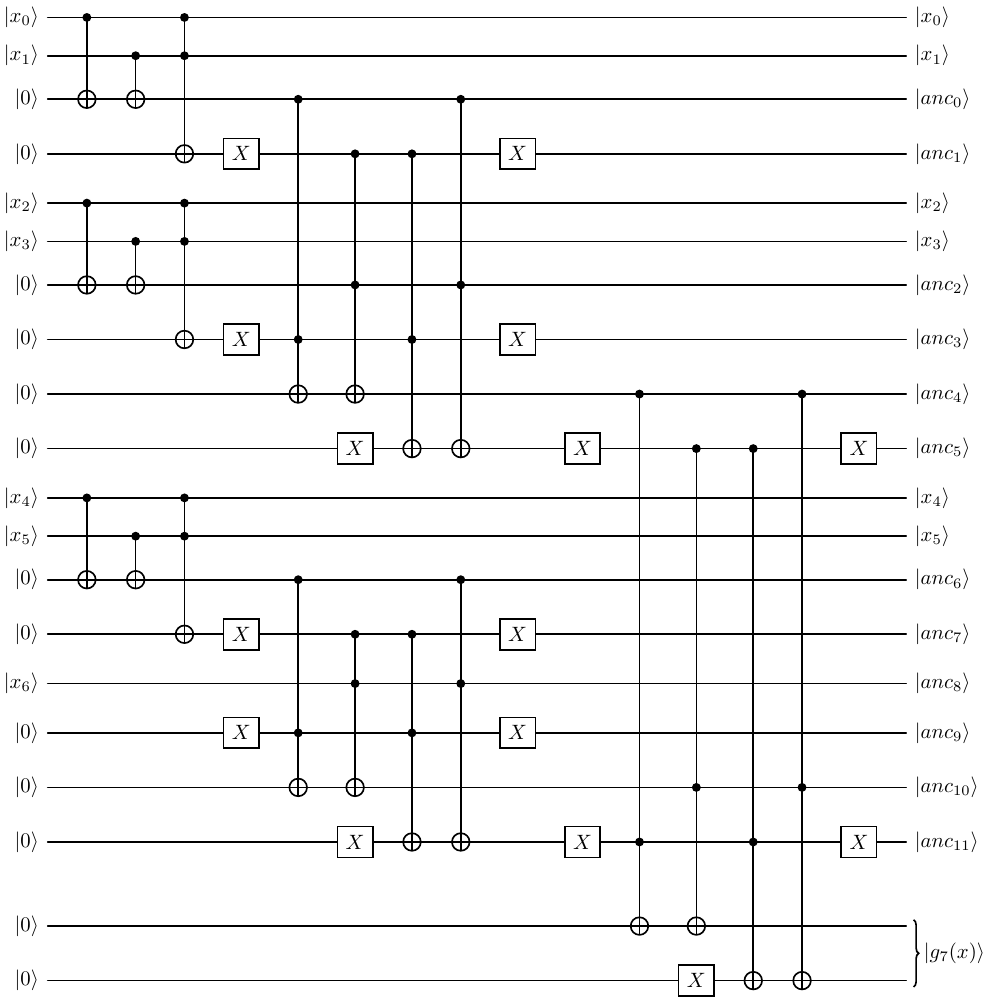}
    \caption{\centering $G_7$ circuit.}
    \label{fig:A_4}
\end{figure}

Mathematically, we can write it as follows: \\

Let $\left| \rho_2 \right\rangle = \left| x \right\rangle \left| 0 \right\rangle ^{\otimes 2} = \left| x_0 \right\rangle \left| x_1 \right\rangle \left| 0 \right\rangle ^{\otimes 2} $ be a 2-bits input state. \\
$G_2 \left| \rho_2 \right\rangle = \left| x \right\rangle \left| x_0 \oplus x_1 \right\rangle \left| x_0 * x_1 \right\rangle = \left| x \right\rangle \left| \mathbf{1}_{\left\{ \sum\limits_{i=0}^{1} x_i = 1 \right\}} \right\rangle \left| \mathbf{1}_{\left\{ \sum\limits_{i=0}^{1} x_i \ge 2 \right\}} \right\rangle = \left| x \right\rangle \left| g_2(x) \right\rangle $, where $\mathbf{1}_{\left\{A\right\}}$ is the indicator function of $A$.\\

If we note $Q_h$, the quantum circuit for h, and $x = x^{(l)} x^{(k)}$ the concatenation of a bitstring $x^{(l)}$ of length $l$ and a bitstring $x^{(k)}$ of length $k$ such that $k+l=n$:

\[
\resizebox{1.09\textwidth}{!}{$
\begin{aligned}
&Q_h \left| g_l(x^{(l)}) \right\rangle \left| g_k(x^{(k)}) \right\rangle \left| 0 \right\rangle^{\otimes 2} 
= Q_h \left| \mathbf{1}_{\left\{ \sum_{i=0}^{l-1} x^{(l)}_i = 1 \right\}} \right\rangle
  \left| \mathbf{1}_{\left\{ \sum_{i=0}^{l-1} x^{(l)}_i \ge 2 \right\}} \right\rangle
  \left| \mathbf{1}_{\left\{ \sum_{i=0}^{k-1} x^{(k)}_i = 1 \right\}} \right\rangle
  \left| \mathbf{1}_{\left\{ \sum_{i=0}^{k-1} x^{(k)}_i \ge 2 \right\}} \right\rangle
  \left| 0 \right\rangle^{\otimes 2} \\
&= \left| g_l(x^{(l)}) \right\rangle \left| g_k(x^{(k)}) \right\rangle 
  \left| \mathbf{1}_{\left\{ \sum_{i=0}^{l-1} x^{(l)}_i = 1 \right\}} 
  \Big( 1 - \mathbf{1}_{\left\{ \sum_{i=0}^{k-1} x^{(k)}_i \ge 2 \right\}} \Big) 
  \oplus 
  \mathbf{1}_{\left\{ \sum_{i=0}^{k-1} x^{(k)}_i = 1 \right\}} 
  \Big( 1 - \mathbf{1}_{\left\{ \sum_{i=0}^{l-1} x^{(l)}_i \ge 2 \right\}} \Big) \right\rangle 
  \left| 1 - \Big( \Big( 1 - \mathbf{1}_{\left\{ \sum_{i=0}^{l-1} x^{(l)}_i \ge 2 \right\}} \Big) 
  \Big( 1 - \mathbf{1}_{\left\{ \sum_{i=0}^{k-1} x^{(k)}_i \ge 2 \right\}} \Big) 
  \oplus \mathbf{1}_{\left\{ \sum_{i=0}^{l-1} x^{(l)}_i = 1 \right\}} \mathbf{1}_{\left\{ \sum_{i=0}^{k-1} x^{(k)}_i = 1 \right\}} \Big) \right\rangle \\
&= \left| g_l(x^{(l)}) \right\rangle \left| g_k(x^{(k)}) \right\rangle 
  \left| \mathbf{1}_{\left\{ \sum_{i=0}^{l-1} x^{(l)}_i = 1 \ \&\ \sum_{i=0}^{k-1} x^{(k)}_i = 0 \right\}} 
  \oplus 
  \mathbf{1}_{\left\{ \sum_{i=0}^{k-1} x^{(k)}_i = 1 \ \&\ \sum_{i=0}^{l-1} x^{(l)}_i = 0 \right\}} \right\rangle 
  \left| 1 - \Big( \mathbf{1}_{\left\{ \sum_{i=0}^{n-1} x_i \le 2 \right\}} 
  \oplus \mathbf{1}_{\left\{ \sum_{i=0}^{n-1} x_i = 2 \right\}} \Big) \right\rangle \\
&= \left| g_l(x^{(l)}) \right\rangle \left| g_k(x^{(k)}) \right\rangle 
  \left| \mathbf{1}_{\left\{ \sum_{i=0}^{n-1} x_i = 1 \right\}} \right\rangle 
  \left| \mathbf{1}_{\left\{ \sum_{i=0}^{n-1} x_i \ge 2 \right\}} \right\rangle 
  = \left| g_l(x^{(l)}) \right\rangle \left| g_k(x^{(k)}) \right\rangle \left| g_n(x) \right\rangle. \\
\end{aligned}
$}
\]

Thus, if we note $\left| \rho \right\rangle = \left| x \right\rangle \left| 0 \right\rangle ^{\otimes 2(n-1) + a_0}$, the initial state including the ancillae and the flag at the end, we have: \\
\[
\resizebox{1.09\textwidth}{!}{$
\begin{aligned}
&G_n \left| \rho \right\rangle = \left| x \right\rangle 
\left(
\prod_{r=1}^{\left\lfloor \log_2(n) \right\rfloor - \left(1-\left\lceil \log_2(n) - \left\lfloor \log_2(n) \right\rfloor \right\rceil \right)}
\left(
\prod_{i=0}^{\left\lfloor n/2^r \right\rfloor -1} 
\left| \mathbf{1}_{\left\{ \sum_{k=2^r i}^{2^r(i+1)-1} x_k = 1 \right\}} \right\rangle
\left| \mathbf{1}_{\left\{ \sum_{k=2^r i}^{2^r(i+1)-1} x_k \ge 2 \right\}} \right\rangle
\right)
\left| R_r \right\rangle^{\otimes a_{r-1} \mathbf{1}_{\left\{ \sum_{k=0}^{r-1} a_k \ge 2 \ \text{or} \ r = 1\right\}}}
\right)
\left| g_n(x) \right\rangle, \\
\end{aligned}
$}
\]
where $ \left| R_r \right\rangle = 
\begin{cases}
\left| \mathbf{1}_{\left\{ \sum\limits_{k = 2^r \lfloor n/2^r \rfloor}^{n-1} x_k = 1 \right\}} \right\rangle 
\left| \mathbf{1}_{\left\{ \sum\limits_{k = 2^r \lfloor n/2^r \rfloor}^{n-1} x_k \ge 2 \right\}} \right\rangle & \text{if $r \ge 2$,} \\
\left| 0 \right\rangle & \text{if $r = 1$.}
\end{cases}
$
\end{proof}

\begin{Prop} Let $n\in\mathbb N$. Define the function $f_n:\{0, 1\}^n\to\{0, 1\}$ by:
\begin{equation}
f_n(x)=\left\{\begin{aligned}
&1&\text{if $\sum_{i=1}^nx_i=1$,}\\
&0&\text{else.}
\end{aligned}\right.
\end{equation}
We can construct a circuit $E_n$ such that $E_n\ket{x}\ket0=\ket{x}\ket{f_n(x)}$ with $\mathcal O(\log(n))$ depth, $\mathcal O(n)$ size using $2$ zeroed ancillae qubits.
\end{Prop}

\begin{proof}
First, note that :
\begin{equation}
f_n(x) \equiv g_n(x) \pmod2.
\end{equation}

The idea is to subdivide bitstring $x$ of length $n$ in 3 parts, to use conditionally clean ancillae~\cite{Khattar2025riseofconditionally}. As $n=3k+r$, $r \in {\{0,1,2\}}$, we can write $x = x^{(k)}x^{(k+l_1)}x^{(k+l_2)}$, the concatenation of bitstrings of length respectively $k$, $k+l_1$ and $k+l_2$, where $l_1 = \mathds{1}_{r \geq 1}(r)$ and $l_2 = \mathds{1}_{r = 2}(r)$. \\ \\
This decomposition enable us to compute $g_k(x^{(k)})$ with $G_k$ circuit : we can use a first zeroed ancilla and the other $2k+r$ qubits from bitstring $x$ as conditionally clean ancillae to get the 2 output qubits and the $2k-4$ or $2k-3$ ancillae needed. We apply a white MCX gate (controlling if qubit in state zero) on $x^{(k+l_1)}x^{(k+l_2)}$, using Khattar \& Gidney decomposition~\cite{Khattar2025riseofconditionally} in $\mathcal O(\log(n))$ depth with the first zeroed ancilla and the target qubit of the circuit as ancillae and the second zeroed ancilla to store the result. If it is $1$, this means that bitstring $x^{(k+l_1)}x^{(k+l_2)}$ has only zeros, and thus computing $G_k$ with these ancillae gives the right output $g_k(x^{(k)})$. If it is $0$,  this means that bitstring $x^{(k+l_1)}x^{(k+l_2)}$ has at least one one and using these qubits as ancillae for $G_k$ will lead to a wrong output. However, we just have to store $f_k(x^{(k)}) \equiv g_k(x^{(k)}) \pmod2$ on the target qubit when $x^{(k+l_1)}x^{(k+l_2)}$ has only zeros. Choosing the first zeroed ancilla as the first target qubit of $G_k$ to store $f_k(x^{(k)}) \equiv g_k(x^{(k)}) \pmod2$ , we just have to apply then a CCX gate with the two zeroed ancillae of the circuit on the target qubit (Note: if $k=1$, the result is already stored on qubit, so the first control will be directly on it. This happens only for $n \leq 5$). We can now uncompute the conditionally clean  ancillae with $G_k^\dagger$ and a MCX gate. \\ 
This procedure will compute $f_k(x^{(k)})$ on flag qubit only if $f_{2k+r}(x^{(k+l_1)}x^{(k+l_2)})=0$. \\ \\
Applying the procedure 3 times, for each part of bitstring $x = x^{(k)}x^{(k+l_1)}x^{(k+l_2)}$ (with the first zeroed ancilla and the conditionally clean ones, there are always enough ancillae to compute $G$. The worst-case scenario occurs when $n = 3k + 1$, $k + 1$ is odd, and we compute $G_{k+1}$. In that case, we require $2k - 1$ zeroed ancillae and 2 target qubits, i.e., $2k + 1$ conditionally clean ancillae, which are available thanks to the two blocks of $k$ qubit wires and the first ancilla), we will have flag qubit in state $\ket{f_n(x)}$ at the end (see Figure~\ref{fig:A_5}).

\begin{figure}[H]
    \centering
    \includegraphics[width=1\linewidth]{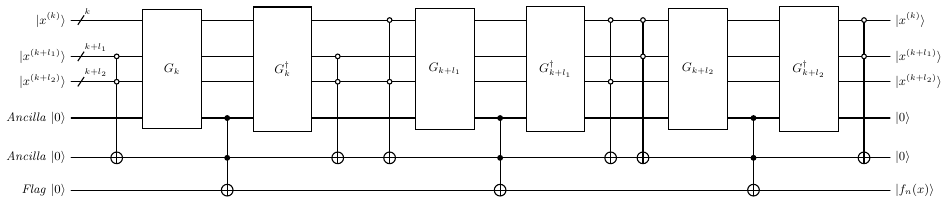}
    \caption{\centering $E_n$ circuit for $n \geq 6$.}
    \label{fig:A_5}
\end{figure}

Mathematically, we can write it as follows: \\

Let $\left| \gamma_0 \right\rangle = \left| x \right\rangle \left| 0 \right\rangle^{\otimes 3} 
= \left| x^{(k)} \right\rangle \left| x^{(k+l_1)} \right\rangle \left| x^{(k+l_2)} \right\rangle \left| 0 \right\rangle^{\otimes 3}
$ 
be the input state, including the two ancillae. \\

After the first MCX, we have: \\ \\
$\left| \gamma_1 \right\rangle = \left| x \right\rangle \left| 0 \right\rangle 
\left| \prod\limits_{i=0}^{k+l_1-1} (1-x_i^{(k+l_1)}) \prod\limits_{i=0}^{k+l_2-1} (1-x_i^{(k+l_2)}) \right\rangle
\left| 0 \right\rangle $ 
\\ \\
Then, applying $G_k$ to the first register $x^{(k)}$, we get:
\[
\resizebox{\textwidth}{!}{$
\left| \gamma_2 \right\rangle = \left| x \right\rangle 
\left| f_k(x^{(k)}) \mathbf{1}_{\left\{ \prod\limits_{i=0}^{k+l_1-1} (1-x_i^{(k+l_1)}) \prod\limits_{i=0}^{k+l_2-1} (1-x_i^{(k+l_2)}) = 1 \right\}} 
+ \sigma \mathbf{1}_{\left\{ \prod\limits_{i=0}^{k+l_1-1} (1-x_i^{(k+l_1)}) \prod\limits_{i=0}^{k+l_2-1} (1-x_i^{(k+l_2)}) = 0 \right\}} \right\rangle
\left| \prod\limits_{i=0}^{k+l_1-1} (1-x_i^{(k+l_1)}) \prod\limits_{i=0}^{k+l_2-1} (1-x_i^{(k+l_2)}) \right\rangle
\left| 0 \right\rangle
$}
\]
where $\sigma$ is an output obtained from $G_k$ when the ancillae are not all zeroed. Its exact value is irrelevant for the next steps. Applying the Toffoli gate gives:
\[
\resizebox{\textwidth}{!}{$
\left| \gamma_3 \right\rangle = \left| x \right\rangle 
\left| f_k(x^{(k)}) \mathbf{1}_{\left\{ \prod\limits_{i=0}^{k+l_1-1} (1-x_i^{(k+l_1)}) \prod\limits_{i=0}^{k+l_2-1} (1-x_i^{(k+l_2)}) = 1 \right\}} 
+ \sigma \mathbf{1}_{\left\{ \prod\limits_{i=0}^{k+l_1-1} (1-x_i^{(k+l_1)}) \prod\limits_{i=0}^{k+l_2-1} (1-x_i^{(k+l_2)}) = 0 \right\}} \right\rangle
\left| \prod\limits_{i=0}^{k+l_1-1} (1-x_i^{(k+l_1)}) \prod\limits_{i=0}^{k+l_2-1} (1-x_i^{(k+l_2)}) \right\rangle
\left| f_k(x^{(k)}) \prod\limits_{i=0}^{k+l_1-1} (1-x_i^{(k+l_1)}) \prod\limits_{i=0}^{k+l_2-1} (1-x_i^{(k+l_2)}) \right\rangle
$}
\]

Uncomputing $G_k$ and the MCX yields:
\[
\resizebox{\textwidth}{!}{$
\left| \gamma_4 \right\rangle = \left| x \right\rangle \left| 0 \right\rangle \left| 0 \right\rangle 
\left| f_k(x^{(k)}) \prod\limits_{i=0}^{k+l_1-1} (1-x_i^{(k+l_1)}) \prod\limits_{i=0}^{k+l_2-1} (1-x_i^{(k+l_2)}) \right\rangle
= \left| x \right\rangle \left| 0 \right\rangle \left| 0 \right\rangle 
\left| f_n(x) \prod\limits_{i=0}^{k+l_1-1} (1-x_i^{(k+l_1)}) \prod\limits_{i=0}^{k+l_2-1} (1-x_i^{(k+l_2)}) \right\rangle
$}
\]
since, if $\prod\limits_{i=0}^{k+l_1-1} (1-x_i^{(k+l_1)}) \prod\limits_{i=0}^{k+l_2-1} (1-x_i^{(k+l_2)}) =1 $, $f_k(x^{(k)}) = f_n(x)$ ; else $\prod\limits_{i=0}^{k+l_1-1} (1-x_i^{(k+l_1)}) \prod\limits_{i=0}^{k+l_2-1} (1-x_i^{(k+l_2)}) =0$. \\

After repeating this for each segment of the bitstring $x$, we finally get:
\[
\resizebox{\textwidth}{!}{$
E_n \left| x \right\rangle \left| 0 \right\rangle^{\otimes 3} = \left| \gamma \right\rangle = 
\left| x \right\rangle \left| 0 \right\rangle \left| 0 \right\rangle 
\left| f_n(x) \prod\limits_{i=0}^{k+l_1-1} (1-x_i^{(k+l_1)}) \prod\limits_{i=0}^{k+l_2-1} (1-x_i^{(k+l_2)}) 
+ f_n(x) \prod\limits_{i=0}^{k-1} (1-x_i^{(k)}) \prod\limits_{i=0}^{k+l_2-1} (1-x_i^{(k+l_2)}) 
+ f_n(x) \prod\limits_{i=0}^{k-1} (1-x_i^{(k)}) \prod\limits_{i=0}^{k+l_1-1} (1-x_i^{(k+l_1)}) \right\rangle
$}
\]

Since 
\[
\resizebox{\textwidth}{!}{$
\prod\limits_{i=0}^{k+l_1-1} (1-x_i^{(k+l_1)}) \prod\limits_{i=0}^{k+l_2-1} (1-x_i^{(k+l_2)}) 
+ \prod\limits_{i=0}^{k-1} (1-x_i^{(k)}) \prod\limits_{i=0}^{k+l_2-1} (1-x_i^{(k+l_2)}) 
+ \prod\limits_{i=0}^{k-1} (1-x_i^{(k)}) \prod\limits_{i=0}^{k+l_1-1} (1-x_i^{(k+l_1)}) 
= 1 - \mathbf{1}_{\left\{ g_k(x^{(k)}) \ge 1 \ \&\ g_{k+l_1}(x^{(k+l_1)}) \ge 1 \ \&\ g_{k+l_2}(x^{(k+l_2)}) \ge 1 \right\}}
$}
\]

and $ f_n(x) \mathbf{1}_{\left\{ g_k(x^{(k)}) \ge 1 \ \&\ g_{k+l_1}(x^{(k+l_1)}) \ge 1 \ \&\ g_{k+l_2}(x^{(k+l_2)}) \ge 1 \right\}} = 0 $,  we finally obtain: \\ 

$ E_n \left| x \right\rangle \left| 0 \right\rangle^{\otimes 3} = \left| \gamma \right\rangle = \left| x \right\rangle \left| 0 \right\rangle \left| 0 \right\rangle \left| f_n(x) \right\rangle$.

\end{proof}

This yields the following corollary.

\begin{Cor}
We can construct a unitary $\Psi_n$ such that $\Psi_n\ket0=\ket{\psi_n}$ in $\mathcal O(\log(n))$ depth with $2$ ancillae qubits.
\end{Cor}

\begin{proof}
Note that $\ket{\psi_n}=E_n\ket{\alpha_n}\ket0$.
\end{proof}

\subsection{Selection operator}
In the following, we show how the unary selection operator can be implemented in logarithmic depth. First, define a flag-controlled version of:
\begin{equation}
\sum_{k=1}^n\ket{1/2^k}\bra{1/2^k}\otimes Z_k+\left(1-\sum_{k=1}^n\ket{1/2^k}\bra{1/2^k}\right)\otimes 1=\prod_{k=1}^n\mathcal C_k(Z_k).
\end{equation}

\begin{Def}[Selection operator]\index{Selection operator} Define the so-called selection operator acting $2n+1$ qubits as follows. First, divide the $2n+1$ qubits into a flag qubit, $n$ control qubits and $n$ target qubits. Then, define
\begin{equation}
S_n=\mathcal C\left(\prod_{j=1}^n\mathcal C_j(Z_j)\right).
\end{equation}
For every $1\leq j\leq n$, $\mathcal C_j(Z_j)$ denotes a $Z$ gate acting on the $j$-th target qubit if and only if the $j$-th control qubit is in state $1$. The outer control is on the state of the flag qubit.
\end{Def}

The following proposition describes a common technique to control unitaries with only a logarithmic depth overhead using ancilla qubits.

\begin{Prop} $S_n$ can be constructed in logarithmic depth using $n$ zeroed ancillae.
\end{Prop}

\begin{proof}
Note that:
\begin{equation}
\left(\prod_{j=1}^n\mathcal C_j(Z_j)\right)=\bigotimes_{j=1}^n\mathcal C(Z).
\end{equation}
To control this operation efficiently, we propagate the state of the flag on $n$ qubits in $\mathcal O(\log(n))$ depth, perform
\begin{equation}
\bigotimes_{j=1}^n\mathcal C(\mathcal C(Z)),
\end{equation}
and restore the ancillae in logarithmic depth.
\end{proof}

\subsection{Symmetric projected unitary encoding of a linear function}

We now show how the Prepare and Select operators can be used to construct a SPUE of $L_n$ in logarithmic depth. We begin by demonstrating that the Prepare and Select operators yield an affine combination of $L_n$.
\begin{Prop}
$(S_n, \ket{\psi_n})$ is a SPUE of 
\begin{equation}
1-2p_nL_n.
\end{equation}
\label{prop:spue_1m2pLn}
\end{Prop}

\begin{proof}
Compute the action of $S_n$ on $\ket{\psi_n}$:
\begin{equation}
S_n\ket{\psi_n}=\sqrt{\frac{p_n}{1-1/2^n}}\sum_{k=1}^n\frac1{\sqrt{2^k}}\ket{1/2^k}\ket 1\otimes Z_k+\sqrt{1-p_n}\ket{\gamma_n}\ket0.
\end{equation}
Project the resulting state on $\ket{\psi_n}$ and use Proposition \ref{prop:pauli_decomposition},
\begin{equation}
\begin{split}
\bra{\psi_n}S_n\ket{\psi_n}&=\frac{p_n}{1-1/2^n}\sum_{k=1}^n\frac{Z_k}{2^k}+1-p_n\\
&=p_n(1-2L_n)+1-p_n\\
&=1-2p_nL_n.
\end{split}
\end{equation}
\end{proof}

\begin{Cor} $(\Psi_n^\dag S_n\Psi_n, \ket0)$ is a SPUE of $1-2p_nL_n$.
\label{cor:spue_zero}
\end{Cor}

\begin{proof}
Recall that, by definition, $\Psi_n\ket0=\ket{\psi_n}$:
\begin{equation}
\braket{0|\Psi_n^\dag S_n\Psi_n|0}=\braket{\psi_n|S_n|\psi_n}=1-2p_nL_n,
\end{equation}
by Proposition \ref{prop:spue_1m2pLn}.
\end{proof}

This translates into a SPUE of a monomial, up to a multiplicative constant. The following theorem is a formal version of Theorem~\ref{the:informal_xketxbrax}

\begin{The}[Formal version of Theorem~\ref{the:informal_xketxbrax}]Let $n\in\mathbb N^*$, $(\ket0\bra0\otimes1+\ket1\bra1\otimes (-\Psi_n^\dag S_n\Psi_n), \ket+\ket0)$ is a SPUE of 
\begin{equation}
p_nL_n.
\end{equation}
Equivalently, the \((n+\mathcal O(n))\)-qubit unitary 
\begin{equation}
\left( H\otimes1 \right) \left(\ket0\bra0\otimes1+\ket1\bra1\otimes (-\Psi_n^\dag S_n\Psi_n)\right)\left( H\otimes1 \right)
\label{eq:spueeq}
\end{equation}
is a block-encoding of $L_n$ with normalization $1/p_n\in \mathcal O(1)$, where \(p_n\) is defined in Equation~\ref{eq:pn}, and can be implemented by a quantum circuit of size \(\mathcal O(n)\) and depth \(\mathcal O(\log (n))\).
\label{the:formal_xketxbrax}
\end{The}

\begin{proof}
Using Corollary \ref{cor:spue_zero}:
\begin{equation}
\begin{split}
\braket{+, 0|\ket0\bra0\otimes1+\ket1\bra1\otimes (-\Psi_n^\dag S_n\Psi_n)|+,0}&=\left\langle0\left|\frac{1-\Psi_n^\dag S_n\Psi_n}2\right|0\right\rangle\\
&=\frac{1+2p_nL_n-1}2\\
&=p_nL_n.
\end{split}
\end{equation}
The circuit size and depth are a direct consequence of the previous propositions and corrolaries.
\end{proof}

\subsection{Arbitrary success probability encoding}
\label{sec:Arbitrary success probability encoding}

When one needs to use $L_n$ many times, one needs to encode it up to a multiplicative constant that is as close as possible to $1$. This can be achieved with logarithmic overhead by using amplitude amplification. Amplitude amplification can be seen as a singular value transform with the polynomial of Lemma~25 in~\cite{qsvt}. 

\begin{Cor} Let $\delta\in ]0, 1[$. We can construct a circuit $\Psi_n^\delta$ using $\mathcal O(\log(1/\delta))$ times the unitary $\Psi_n$ such that:
\begin{equation}
\exists\epsilon\in [0, \delta]:\Psi_n^\delta\ket0=\sqrt{1-\epsilon}\ket{\beta_n}\ket1+\sqrt{\epsilon}\ket{\gamma_n^\delta}\ket0,
\end{equation}
for some state $\left|\gamma_n^\delta\right\rangle$ that is orthogonal to $\ket{\beta_n}$.
\end{Cor}

\begin{Cor}
$\left({\Psi_n^\delta}^\dag S_n\Psi_n^\delta, \ket0\right)$ is a SPUE of $1-2(1-\epsilon)L_n$.
\end{Cor}

\begin{Cor} $\left(\ket0\bra0\otimes1+\ket1\bra1\otimes \left(-{\Psi_n^\delta}^\dag S_n\Psi_n^\delta\right), \ket+\ket0\right)$ is a SPUE of 
\begin{equation}
(1-\epsilon)L_n.
\end{equation}
\label{cor:1mepsLn}
\end{Cor}

Depending on the user preferences, one can also performed directly the amplitude amplification on $\Psi_n^\dag S_n\Psi_n$, preparing similarly a SPUE of $(1-\epsilon)L_n$ with an overheard $\mathcal{O}(\log(\delta))$ in complexity, where $\epsilon \leq \delta$.

\section{Polynomials of the position operator and state preparation}

We can now obtain polynomial diagonal operators from SPUE's of $L_n$. In the present appendix, we only consider polynomials $p$ such that $\max_{x\in [-1, 1]}|p(x)|\leq 1$\ \footnote{This condition can be removed by simply replacing each occurrence of $p$ in the following computations with $p/\|p\|_{\infty}$.}.

\begin{Prop}
Let $\eta>0$ and $p$ be a degree-$d$ complex polynomial with scaling factor $\beta$ and such that $\max_{x\in [-1, 1]}|p(x)|\leq 1$. Then, we can construct a PUE 
\begin{equation}\left(\ket0\bra0\otimes1-\ket1\bra1\otimes \left(\Psi_n^{\eta/\|p'\|_\infty}\right)^\dag S_n\Psi_n^{\eta/\|p'\|_\infty}, \ket+\ket0\ket0, \ket+\ket0\ket0\right)
\end{equation}
of
\begin{equation}
\frac1\beta p(L_n)=\frac1\beta\sum_{x\in B_n}p(x)\ket x\bra x,
\end{equation}
up to spectral norm error $\eta/\beta$ with depth $\mathcal O(d\log(\|p'\|_\infty/\eta)\log(n))$. Here $\|\cdot\|_\infty$ is the maximum norm on $[-1, 1]$
\end{Prop}

\begin{proof}
Let $\delta>0$. Corollary~\ref{cor:1mepsLn} provides a SPUE of $(1-\epsilon)L_n$, with $0\leq \epsilon<\delta$ and depth $\mathcal O(\log(1/\delta)\log(n))$. Using this $SPUE$ $d$ times in the construction of Theorem~\ref{the:gqet}, we obtain a PUE of
\begin{equation}
\frac1\beta p((1-\epsilon)L_n).
\end{equation}
By Rolle's Theorem, for every $x\in B_n$:
\begin{equation}
\exists y\in [(1-\epsilon)x, x]:p((1-\epsilon)x)-p(x)=\epsilon xp'(y).
\end{equation}
Therefore, 
\begin{equation}
\|p((1-\epsilon)L_n)-p(L_n)\|\leq \epsilon\|p'\|_\infty\leq \delta\|p'\|_\infty.
\end{equation}
Taking $\delta=\eta/\|p'\|_\infty$ yields the claim.
\end{proof}

\begin{The}
\label{thm:poly_formal}
Let $n\in\mathbb N^*$, $\{ a_k\}_{k=0}^d\in \mathbb C^{d+1}$, $\epsilon>0$ and $p$ be a $d$-degree polynomial such that $\forall x \in \mathbb C, p(x)=\sum_{k=0}^da_kx^k$ and $\max_{x\in[0,1]}|p(x)|\leq 1$.  One can prepare the state 
\begin{equation}
\frac1{\|p\|_{2,N}}\sum_{x\in B_n}p(x)\ket x,
\end{equation}
using a quantum circuit of logarithmic depth $\mathcal O(d\log(\frac{\|p'\|_2}{\|p\|_2}/\epsilon) \log(n))$, size $\mathcal O(nd\log(\frac{\|p'\|_2}{\|p\|_2}/\epsilon))$ and $\mathcal O(n)$ ancillae, with a
 success probability $\mathbb P =\left\|\frac1{\beta2^{n/2}}\sum_{x\in B_n}p(x)\ket x\right\|^2\to\frac1{\beta^2}\int_0^1p(x)^2dx.$ where $\beta$ is the scaling factor of the polynomials defined in Equation \ref{def:scalingfac}.
\label{thm:maintheorem}
\end{The}

\begin{proof}
    The prepared quantum state after success is $\ket{p_\epsilon}=(1/\|p_\epsilon\|_{2,N})\sum_{x\in B_n}p((1-\epsilon)x)\ket x$. Thus,  the difference with the target quantum state $\ket{p}=(1/\|p\|_{2,N})\sum_{x\in B_n}p(x)\ket x$ can be bounded as
\begin{equation}
\label{eq:proof of thm 7}
\begin{split}
    \|\ket{p}-\ket{p_\epsilon}\|_{2,N}&\leq \|\ket{p}-\frac{1}{\|p\|_{2,N}}\sum_{x\in B_n}p((1-\epsilon)x)\ket x\|_{2,N}+\|\frac{1}{\|p\|_{2,N}}\sum_{x\in B_n}p((1-\epsilon)x)\ket x-\ket{p_\epsilon}\|_{2,N} \\&  \leq \frac{1}{\|p\|_{2,N}}\sqrt{\sum_{x\in B_n} |p(x)-p((1-\epsilon)x)|^2 }+\frac{|\|p_\epsilon\|_{2,N}-\|p\|_{2,N} |}{\|p\|_{2,N}}  
    \\&\leq \frac{2}{\|p\|_{2,N}}\sqrt{\sum_{x\in B_n} |p(x)-p((1-\epsilon)x)|^2 }\leq\frac{2\epsilon}{\|p\|_{2,N}}\sqrt{\sum_{x\in B_n} |p'(y_x)|^2}\xrightarrow{N\rightarrow +\infty} 2\epsilon \|p'\|_2/\|p\|_2
\end{split}
\end{equation}
where the triangular inequality and the subadditivity of the norm have been used. One can choose $\epsilon=O(\epsilon'\|p\|_2/\|p'\|_2)$ to obtain an $\epsilon'$ approximation of the target state. 

Additionally, the success probability is $\mathbb P =\left\|\frac1{\beta2^{n/2}}\sum_{x\in B_n}p((1-\epsilon)x)\ket x\right\|^2\xrightarrow{N\rightarrow ±\infty,\epsilon\rightarrow 0}\frac1{\beta^2}\int_0^1|p(x)|^2dx$ where $\beta$ is the scaling factor of the polynomial $p$ defined Equation \ref{def:scalingfac}.
Therefore, by combining the SPUE of $p_nL_n$, the amplitude amplification step, and the GQET step, the overall quantum circuit has a depth $\mathcal O(d\log(\frac{\|p'\|_2}{\|p\|_2}/\epsilon)\log(n))$, a size $\mathcal O(nd\log(\frac{\|p'\|_2}{\|p\|_2}/\epsilon))$, and uses $\mathcal O(n)$ ancillae.
\end{proof}

\section{Quantum state preparation with approximate GQET protocol}
\firstreview{Theorem \ref{thm:poly_formal} assumes that the GQET protocol implements the target polynomial $P$ exactly. In practice, to determine the angles of the GQET protocol, one must first determine the complementary polynomial $Q$ such that $\forall z\in \mathbb{T}, |P(z)|^2+|Q(z)|^2=1$. Several approaches have been proposed, including root finding \cite{qsvt}, Prony's method \cite{ying2022stable}, optimization \cite{gqsp,dong2021efficient}, or contour integral representation \cite{berntson2025complementary}. These numerical procedures generally yield approximate angles, introducing an additional error in the resulting quantum state.}

\firstreview{In the following, we show that this additional error can be bounded directly in terms of the errors in the individual angles.}

\firstreview{Consider the notations of Theorem \ref{thm:poly_formal}, a set of angles $\{ \theta_j,\phi_j,\lambda\}$ associated with a unitary $\widehat{\text{GQET}}_{p_\epsilon}$ which is a block-encoding of the diagonal operator $\sum_{x\in B_n}p((1-\epsilon)x)\ket{x}\bra{x}$, and a set of angles obtained through an angle-finding algorithm $\{ \tilde{\theta}_j,\tilde{\phi}_j,\tilde{\lambda}\}$ associated with a unitary $\widehat{\text{GQET}}_{q_\epsilon}$ which is a block-encoding of $\sum_{x\in B_n}q((1-\epsilon)x)\ket{x}\bra{x}$. The polynomial $q$ corresponds to the implemented approximation of $p$.
Then, the difference between the target quantum state $\ket{p}$ and the prepared one $\ket{q_\epsilon}=(1/\|q_\epsilon\|_{2,N})\sum_{x\in B_n}q((1-\epsilon)x)\ket x$ can be bounded as
\begin{equation}
\begin{split}
    \|\ket{p}-\ket{q_\epsilon}\|_{2,N} \leq \|\ket{p}-\ket{p_\epsilon}\|_{2,N}+\|\ket{q_\epsilon}-\ket{p_\epsilon}\|_{2,N}
\end{split}
\end{equation}
The first term is bounded in Equation \ref{eq:proof of thm 7}. Then, notice that $\ket{p_\epsilon}=\frac{\sqrt{N}}{\|p_\epsilon\|_{2,N}} p((1-\epsilon)L_n)H^{\otimes n}\ket{0}^{\otimes n}$, $\ket{q_\epsilon}=\frac{\sqrt{N}}{\|q_\epsilon\|_{2,N}}q((1-\epsilon)L_n)H^{\otimes n}\ket{0}^{\otimes n}$ and, by definition of the block encodings, $p((1-\epsilon)L_n)H^{\otimes n}\ket{0}^{\otimes n}=\bra{0}^{\otimes m} \widehat{\text{GQET}}_{p_\epsilon}H^{\otimes n}\ket{0}^{\otimes n}\ket{0}^{\otimes m} $, $q((1-\epsilon)L_n)H^{\otimes n}\ket{0}^{\otimes n}=\bra{0}^{\otimes m} \widehat{\text{GQET}}_{q_\epsilon}H^{\otimes n}\ket{0}^{\otimes n}\ket{0}^{\otimes m} $, where $m$ is the number of ancilla qubits associated with the GQET protocol. The second term can be bounded as
\begin{equation}
\begin{split}
    \|\ket{q_\epsilon}-\ket{p_\epsilon}\|_{2,N}&\leq \frac{2\sqrt{N}}{\|p_\epsilon\|_{2,N}}\|q((1-\epsilon)L_n)H^{\otimes n}\ket{0}^{\otimes n}-p((1-\epsilon)L_n)H^{\otimes n}\ket{0}^{\otimes n} \|_{2,N} \\
    &=\frac{2\sqrt{N}}{\|p_\epsilon\|_{2,N}}\|\ket{0}^{\otimes m}\bra{0}^{\otimes m}  \left( \widehat{\text{GQET}}_{q_\epsilon}H^{\otimes n}\ket{0}^{\otimes n}\ket{0}^{\otimes m}-\widehat{\text{GQET}}_{p_\epsilon}H^{\otimes n}\ket{0}^{\otimes n}\ket{0}^{\otimes m} \right) \|_{2,N+M} \\
    & \leq \frac{2\sqrt{N}}{\|p_\epsilon\|_{2,N}}\|\left( \widehat{\text{GQET}}_{q_\epsilon}-\widehat{\text{GQET}}_{p_\epsilon} \right)H^{\otimes n}\ket{0}^{\otimes n}\ket{0}^{\otimes m} \|_{2,N+M} \\
    & \leq \frac{2\sqrt{N}}{\|p_\epsilon\|_{2,N}}\|\widehat{\text{GQET}}_{q_\epsilon}-\widehat{\text{GQET}}_{p_\epsilon}  \|_{2} 
\end{split}
\end{equation}
where the term $\frac{\sqrt{N}}{\|p_\epsilon\|_{2,N}}$ tends to the constant $\frac{1}{\|p\|_{L_2}}$ in the limit $N\rightarrow +\infty$ and $\epsilon\rightarrow0$. As shown in Figure \ref{fig:quantum_circuit}, each GQET operation consists of a sequence of rotations $R_j=R(\theta_j,\phi_j,\lambda_j)$, where $\lambda_j=0$ for $j>0$ and $\lambda_0=\lambda$, alternating with a unitary $V$, i.e., $\widehat{\text{GQET}}_{p_\epsilon}=(\prod_{j=1}^d (R_j V))R_0$ and $\widehat{\text{GQET}}_{q_\epsilon}=(\prod_{j=1}^d (\tilde{R}_j V))\tilde{R}_0$, where the $\tilde{R}_j$ are associated with the $\{ \tilde{\theta}_j,\tilde{\phi}_j,\tilde{\lambda}\}$ angles. Using the fact that $\|U_1...U_d-V_1...V_d\|_2\leq \sum_{j=1}^d \|U_j-V_j\|_2$ for any set of unitaries, and $\|R(\theta,\phi,\lambda)-\tilde{R}(\tilde{\theta},\tilde{\phi},\tilde{\lambda})\|_2\leq |\theta-\tilde{\theta}|+|\phi-\tilde{\phi}|+|\lambda-\tilde{\lambda}|$, one can directly prove the following:
\begin{equation}
\begin{split}
\|\widehat{\text{GQET}}_{q_\epsilon}-\widehat{\text{GQET}}_{p_\epsilon}  \|_{2}& = \|\left(\prod_{j=1}^d (\tilde{R}_j V)\right)\tilde{R}_0-\left(\prod_{j=1}^d (R_j V)\right)R_0\|_{2}\\
& \leq \sum_{j=0}^d\| \tilde{R}_j-R_j\|_2
     \leq  |\lambda-\tilde{\lambda}| + \sum_{j=0}^d |\theta_j-\tilde{\theta}_j|+|\phi_j-\tilde{\phi}_j|. 
\end{split}
\end{equation}
}
\\

\firstreview{Thus, the additional error introduced by the angle-finding procedure is directly controlled by its numerical accuracy. These algorithms have been demonstrated to achieve errors close to classical machine precision for polynomials of very high degree \cite{dong2021efficient}.
}

\section{\firstreview{Approximation of monomials by polynomials and polynomials of sines}}

\firstreview{In this Appendix, an example is presented in which the method proposed in this work outperforms the approach based on $\sin^{-1}$ introduced in \cite{mcardle2022quantum}. As illustrated in the following Figure~\ref{fig:log-Max-Error_Versus_Degree}, achieving a specific target precision requires a significantly lower polynomial degree when approximating $x^{20}$ compared to $(\arcsin(x))^{20}$. The logarithmic maximum error for $x^{20}$ decreases much more rapidly, demonstrating that it can be represented efficiently by lower-degree polynomials, whereas the more complex behavior of $(\arcsin(x))^{20}$ demands higher-degree terms to achieve the same level of accuracy.}

\begin{figure}[H]
\centering
\includegraphics[width=.7\linewidth]{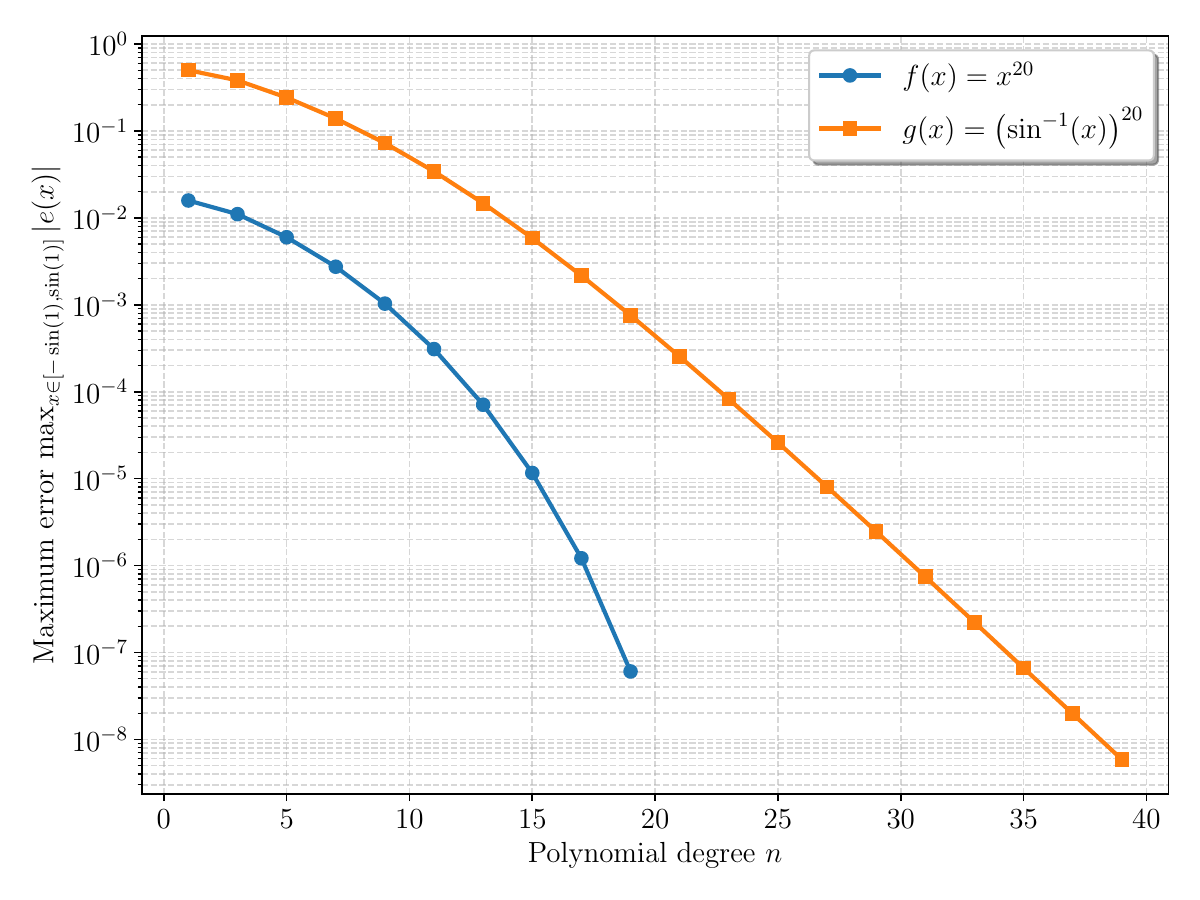}
\caption{\centering Logarithmic maximum error of the polynomial approximations for $f(x) = x^{20}$ and $g(x) = (\arcsin(x))^{20}$ as a function of the polynomial degree $n$. The maximum error is evaluated as $\max_{x \in [-\sin(1), \sin(1)]} |e(x)|$, where $e(x)$ represents the approximation error (i.e., the difference between the target function and its corresponding polynomial approximation of degree $n$). As illustrated, achieving a specific target precision requires a significantly lower polynomial degree when approximating $x^{20}$ compared to $(\arcsin(x))^{20}$.}    \label{fig:log-Max-Error_Versus_Degree}
\end{figure}

\end{document}